\documentclass[a4paper]{article}
\usepackage[UKenglish]{babel}
\usepackage{hyperref}
\usepackage{thm-restate}
\usepackage{authblk}
\usepackage{amsthm}
\usepackage{thmtools}

\newtheorem{theorem}{Theorem}[section]
\newtheorem{lemma}[theorem]{Lemma}
\newtheorem{definition}[theorem]{Definition}
\newtheorem{proposition}[theorem]{Proposition}
\newtheorem{observation}[theorem]{Observation}
\newtheorem{corollary}{Corollary}[section]

\usepackage{amsmath, amssymb, bussproofs, cancel, url, latexsym, multirow, graphicx, stmaryrd, booktabs, float, tikz, listings, forest, tasks, multicol, tabularx, xspace}
\usepackage{cleveref}

\usetikzlibrary{arrows.meta,decorations.pathmorphing}

\newcommand\magicwand{\mathrel{-\mkern-6mu*}}

\newcommand{\septraction}{%
  \mathrel{%
    \begin{tikzpicture}[baseline=-0.6ex]
      \draw (0,0) -- (0.9em,0);

      \node[
        draw,
        circle,
        fill=white,          
        minimum size=0.9ex,  
        inner sep=0pt,
        line width=0.1pt
      ] at (0.9em,0) {$\ast$};
    \end{tikzpicture}%
  }%
}

\newcommand{\BBI}{\textbf{BBI}\xspace}
\newcommand{\SFour}{\textbf{S4}\xspace}
\newcommand{\rev}[1]{[#1]'}

\title{Embedding Modal Logics into Logics of Bunched Implications
} 
\author{%
  Daniele Sansoni\\Australian National University\\\texttt{daniele.sansoni@anu.edu.au}
  \and
  Ranald Clouston\\Australian National University\\\texttt{ranald.clouston@anu.edu.au}
}

\date{}
\begin{document}

\maketitle

\begin{abstract}
We present a new proof of the embedding of the classical modal logic S4 into the logic of Boolean Bunched Implications (BBI).
While the original proof is semantical, this proof is entirely syntactical.
It proceeds via Hilbert-style calculi,
and is built by analogy with a recently discovered proof by G\"odel of the embedding of intuitionistic propositional logic into S4.
We present the first full proofs of deduction theorems for BBI,
which are used to show that the embedding is preserved by reasoning with assumptions, including where those assumptions are organised into bunches.
Unlike the existing proof, our proof is stable under arbitrary axiomatic extensions of S4, and applies to all known axiomatic extensions of BBI in the literature. We observe how this embedding is related to semantical properties of both logics. Moreover, the proof extends gracefully to language extensions of BBI, as we show with hybrid BBI, classical BI, and sub-classical BBI.
\end{abstract}

\section{Introduction}

The Logic of Bunched Implications (BI)~\cite{pym_semantics_2002} is the combination of the usual `additive' connectives of propositional logic with `multiplicative' connectives that need not admit weakening and contraction.
Although its original development~\cite{OHearn1999Logic} was influenced by linear logic~\cite{girard1987Linear}, BI differs by lacking exponentials, and by additive implication being a first class connective.
We will concentrate on Boolean BI (\BBI), along with some extensions, in which the additive connectives are classical.
\BBI is strictly more expressive than its intuitionistic counterpart~\cite{Ishtiaq2001BI}, and forms the propositional basis of most separation logics \cite{1029817}, which are Hoare-style logics originally intended for reasoning about programs that explicitly mutate memory.
Separation logic has been used for program verification in too many ways to sensibly summarise in this paper, but we note that its diverse applications create demand for general techniques that work across the family of separation logics \cite{Parkinson2010Next,Bringing_order}.
Similarly, general techniques are desirable for its foundational logic.

Logics of bunched implications are intended for reasoning about \textit{resources}, and their splitting and combination.
$\alpha * \beta$, where $*$ is often called \emph{separating conjunction}, can be read as ``available resources may be split into a part that makes $\alpha$ true and a part that makes $\beta$ true''.
$\alpha \magicwand \beta$, where $\magicwand$ is often called \emph{magic wand}, can be read as ``given any resources that make $\alpha$ true, if they can be combined with our currently available resources then their combination makes $\beta$ true''.
The archetypal example is fragments of the memory heap, which may be combined if they assign disjoint locations.
Axiomatic and language extensions of \BBI have been proposed for a number of reasons.
First, while certain classes of models, of interest in separation logic applications, have been found to correspond to axiomatic extensions of \BBI, many other such classes cannot be so specified, and require 
a \textit{hybrid logic} extension of the language of \BBI~\cite{separation_theories}.
Second, modal \textit{epistemic} extensions of \BBI have been proposed and used to reason about access control \cite{courtault:hal-01259768,10.1093/logcom/exz024}. 
Third, there has been interest in extending the language of \BBI with a more complete suite of multiplicative propositional connectives~\cite{Brotherston2010ClassicalBI,intermediateBBI}.

This paper looks at \emph{embeddings} from classical modal logics into logics of bunched implications.
The study of such relationships between logics, \emph{interpretational proof theory}, is one of the major branches of proof theory~\cite{Troelstra2000Basic}. Examples of its utility include the embedding of intuitionistic and intermediate logics into their modal companions \cite{chagrov_modal_1992}, with applications to proofs of decidability, Kripke completeness, and the finite model property. This embedding is \emph{parametric}, in the sense that it is stable under axiomatic extension: for a translation $t(-)$ from logic $S1$ to logic $S2$, we have that $\vdash_{S1+\beta}\alpha$ iff $\vdash_{S2+t(\beta)} t(\alpha)$.
Related to this line of work is Artemov's ``realisation'' (a Skolemization-like technique substituting occurrences of $\Box$ with structured terms and functions) of \SFour into his ``logic of proof'' \textbf{LP}, ultimately providing an interpretation of BHK semantics in a formal classical environment \cite{artemov2019justification}.

Combining a number of results in the literature \cite{chagrov_modal_1992, artemov2019justification, embedding_BI_into_BBI, galmiche_expressivity_2006}, \BBI can be placed in this wider network of translations between systems, as shown in Figure~\ref{fig:embeddings}. The embedding of \SFour into \BBI was proved by Galmiche and Larchey-Wendling~\cite{galmiche_expressivity_2006}. Yet this result is more fragile than the ones connecting intermediate, modal, and justification logics. Faithfulness of the embedding is proved \cite{galmiche_expressivity_2006} by a semantic argument: we can always produce a tree such that if a formula $\alpha$ is not \SFour-provable, then the least element of the tree $r$ is such that $r\nvDash \alpha$. This tree is then used to build a \textbf{BBI} resource frame. This proof cannot work for all axiomatic extensions of \textbf{S4}, because not all such extensions have semantics with the tree property. Moreover, no observation is made by Galmiche and Larchey-Wendling regarding whether the embedding is stable under reasoning with assumptions, and it is unclear to us whether their methodology could extend even to this result.

\begin{figure}[t]
\centering
\begin{tikzpicture}[
    >=Stealth,
    node style/.style={
        draw,
        circle,
        minimum size=1.1cm,
        font=\small,
        thick
    }
]

% Nodes
\node[node style] (CPC) at (-3,0) {CPC};
\node[node style] (IPC) at (0,0) {IPC};
\node[node style] (S4)  at (3,0) {S4};
\node[node style] (LP)  at (6,0) {LP};

\node[node style] (BI)  at (0,-2) {BI};
\node[node style] (BBI) at (3,-2) {BBI};

% Arrows
\draw[->, thick, dash pattern=on 1pt off 3pt] (S4) -- (BBI);
\draw[->, thick, decorate, decoration={snake, amplitude=0.75 pt, segment length=7pt}] (S4) -- (LP);
\draw[->, thick] (IPC) -- (S4);
\draw[->, thick, dash pattern=on 1pt off 3pt] (BI) -- (BBI);
\draw[->, thick] (CPC) -- (IPC);

\end{tikzpicture}
\caption{A visualisation of a network of embeddings. \\
\footnotesize Unbroken arrows indicate that there is a known parametric embedding. The embedding of classical propositional logic \textbf{CPC} into \textbf{IPC} is trivially parametric, as classical logic admits no consistent axiomatic extension. The embedding of \textbf{IPC} into \SFour is parametric through the notion of modal companions \cite{chagrov_modal_1992}. Dotted arrows indicate embeddings that are not known to be parametric, of \textbf{BI} into \BBI~\cite{embedding_BI_into_BBI}, and of \SFour into \BBI~\cite{galmiche_expressivity_2006}. A contribution of this paper is to develop the latter embedding into a parametric one. The wavy arrow is for the related but more complex notion of realisation, connecting \SFour to \textbf{LP}; this relationship has been extended to identify justification counterparts of the systems of the normal modal cube~\cite{artemov2019justification, borg2015realization}.}
\label{fig:embeddings}
\end{figure}

In this paper, we remedy this frailty. The methodology we employ in proving the embedding of \SFour into \BBI is parametric, and purely syntactical, allowing us to establish a general result linking normal modal systems extending \SFour to extensions of \BBI (their ``bunched companions'').  Our methodology is of historical interest, as it adapts techniques from a recently discovered alternative proof by G\"odel of his well known embedding of \textbf{IPC} into \SFour. Negri and Von Plato~\cite{Godel_Grundlagen} describe these techniques they unearthed as ``previously unknown''. One of the contributions of this paper is to show that G\"odel's long forgotten techniques have utility beyond the logics they were originally applied to.

The paper is organised as follows. In Section~\ref{Sec:prelim}, we present Hilbert-style systems and relational semantics for \SFour and \BBI. In Section~\ref{Sec:embedding}, we prove the faithfulness of the embedding of \SFour into \BBI using only syntactical means. In Section~\ref{Sec:Assumptions}, we provide extensions of the Hilbert systems for both \SFour and \BBI supporting reasoning with assumptions. We then prove deduction theorems for these systems, on both $\rightarrow$ and $\magicwand$, and that the embedding of \SFour into \BBI holds under reasoning with assumptions. In Section~\ref{Sec:general} we extend our results to a general statement concerning axiomatic extensions of both \SFour and \BBI. This allows us to connect our syntactic results to relational semantics, showing that completeness for a number of interesting \BBI-extensions may be obtained. In Section~\ref{Sec:language_ext}, we observe how the embedding result remains stable under several interesting language extensions of \BBI. In Section~\ref{Sec:conc}, we conclude by discussing related work, and potential future directions of our research.

\section{Syntax and Semantics for \BBI and \SFour}
\label{Sec:prelim}

The logic of Boolean bunched implications (\BBI) has the usual classical operators ($\lor$, $\land$, $\rightarrow$, $\neg$, $\top$, $\bot$) as well as multiplicative symbols $*$, $I$, $\magicwand$.
We occasionally use the symbol $\septraction$ and the name \emph{septraction}
for the De Morgan dual of $\magicwand$,
defined as $\alpha \septraction \beta \equiv \neg(\alpha \magicwand \neg \beta)$~\cite{Calcagno:Modular}.
We adopt the convention that $\neg$ binds more tightly than $\lor$, $\land$, $*$, which in turn bind more tightly than $\rightarrow$, $\magicwand$, $\septraction$.
We present a standard (\cite{galmiche_expressivity_2006, separation_theories, intermediateBBI}) Hilbert-style axiomatic system for \BBI: 

\begin{definition}[Axiomatic System \BBI] Add to an axiomatisation of classical propositional logic the following axiom schemes:

\begin{enumerate}
\begin{multicols}{2}
\item [B1.] $\vdash \phi \rightarrow (I*\phi)$.
\item[B2.] $\vdash (I*\phi) \rightarrow \phi$.
\item[B3.] $\vdash (\phi*\psi)\rightarrow(\psi*\phi)$.
\item[B4.] $\vdash (\phi*(\psi*\chi)) \rightarrow ((\phi*\psi)*\chi))$.

\end{multicols}

\end{enumerate}

We have the following deduction rules:

\begin{enumerate}
\begin{multicols}{2}

\item[R1.] \AxiomC{$\alpha$}
\RightLabel{[MP]}
\AxiomC{$\alpha \rightarrow \beta$}
\BinaryInfC{$\beta$}

\DisplayProof

\vspace{0.5em}
\item[R2.] \AxiomC{$\alpha \rightarrow (\beta \magicwand \gamma)$}
\RightLabel{[$\magicwand 1$]}
\UnaryInfC{$(\alpha * \beta)\rightarrow \gamma$}

\DisplayProof

\item[R3.] \AxiomC{$(\alpha * \beta)\rightarrow \gamma$}
\RightLabel{[$\magicwand 2$]}
\UnaryInfC{$\alpha \rightarrow (\beta \magicwand \gamma)$}

\DisplayProof

\vspace{0.5em}
\item[R4.] \AxiomC{$(\alpha \rightarrow \gamma)$}
\AxiomC{$(\beta \rightarrow \delta)$}

\RightLabel{[$*$]}
\BinaryInfC{$( \alpha * \beta)\rightarrow (\gamma * \delta)$}

\DisplayProof

\end{multicols}

\end{enumerate}
\end{definition}

Rules R2 and R3 deal with the import/export of multiplicative conditional $\magicwand$ in relation to multiplicative conjunction $*$, while rule R4 states that $*$ behaves like additive conjunction $\land$ when joining antecedents and consequents of additive conditional. 

System \BBI has been proven sound and complete with respect to relational semantics based on partial non-deterministic monoids \cite{galmiche_expressivity_2006}. The term ``non-deterministic'' comes from the fact that the composition of elements of the monoid $r_1$ and $r_2$ may yield not one but several results (a set of results), including the possible incompatibility of $r_1$ and $r_2$ (i.e. the empty set). This relational semantics provides the intuitive resource-based interpretation of the substructural operators \cite{pym_semantics_2002}. We briefly recount the semantics, and provide forcing clauses for the operators:

\begin{definition}[Relational Resource Frame]
    A relational frame on a partial non-deterministic monoid is a triple $(M, \triangleright, e)$, where $e \in M$ and $\triangleright$ is a ternary relation on $M$ such that for all $r_1,r_2,r_3,r_4 \in M$:

    \begin{description}
        \item[Identity] $e,r_1 \triangleright r_2$ iff $r_1=r_2$.
        \item[Commutativity] $r_1,r_2 \triangleright r_3$ iff $r_2,r_1 \triangleright r_3$.
        \item[Associativity] if $\exists r'(r_2,r_3 \triangleright r'$ and $r_1,r' \triangleright r_4)$, then $\exists r''(r_1,r_2 \triangleright r'' $ and $r'', r_3 \triangleright r_4) $.

    \end{description}

\end{definition}

As noted \cite{galmiche_expressivity_2006}, the relation ``$r_1, r_2\triangleright r_3$'' can be intuitively read as follows: ``the composition of $r_1$ and $r_2$ yields $r_3$'', or ``$r_3$ is decomposable in $r_1$ and $r_2$''. Axiomatic system \BBI is not able to account \cite{separation_theories} for all the different features of some classes of models. From this standpoint, we may have chosen a model for which instead of a single empty resource $e$ we could have had a set of empty resources $e_1,e_2,..., e_n \in E$. Still, there are certain features of the monoidal semantics that may actually be axiomatized within \BBI: for example, formula $I\land((\alpha*\beta)\magicwand \bot) \rightarrow ((\alpha \magicwand \bot)\lor (\beta \magicwand \bot))$ is valid only \cite{total_monoid_signature} in total monoids. We now define the relational resource model for \BBI:

\begin{definition}[Relational Resource Model] Let $(M, \triangleright, e)$ be a relational resource frame based on a partial non-deterministic monoid, and $v:Prop \longrightarrow \mathcal{P(M)}$ an evaluation mapping propositional letters to subsets of $M$. Then, we define the forcing relation ``$\Vdash$'' as follows:

    \begin{description}
        \item $r \Vdash p$ iff $r \in v(p)$.
        \item $r \Vdash I$ iff $r=e$.
        \item $r \Vdash \bot$ never.
        \item $r \Vdash \top$ always.
        \item $r \Vdash \neg \phi$ iff $\cancel\Vdash \phi$.
        \item $r \Vdash \phi \land \psi$ iff $r \Vdash \phi$ and $r \Vdash \psi$.
        \item $r \Vdash \phi \lor \psi$ iff $r \Vdash \phi$ or $r \Vdash \psi$.
        \item $r \Vdash \phi \rightarrow \psi$ iff $r \cancel\Vdash \phi$ or $r \Vdash \psi$.
        \item $r \Vdash \phi \magicwand \psi$ iff $\forall r_1,r_2 (r,r_1 \triangleright r_2$ and $r_1 \Vdash \phi)$ then $r_2 \Vdash \psi$.
        \item $r \Vdash \phi * \psi$ iff $\exists r_1,r_2 $ such that $ (r_1,r_2 \triangleright r$, and $r_1 \Vdash \phi$ and $r_2 \Vdash \psi$).

    \end{description}

If $r \Vdash \alpha$ for all $r$ in the model, we write $\vDash_{BBI}\alpha$.
    
\end{definition}

Soundness and completeness of \BBI for the relational model are proven in \cite{galmiche_expressivity_2006}. The resource reading of operators suggested by the relational framework comes in handy in providing an intuitive understanding of the main feature of the translation $t(-)$ used for proving the embedding of \SFour into \BBI: that is, modulo the recursive application of $t(-)$ to sub-formulae, transforming occurrences of ``$\Box \alpha$'' with ``$\top \magicwand \alpha$''. The understanding of $\Box \alpha$ provided by $t(-)$ may be informally considered to state: ``if we provide enough resources to make $\top$ true (that is, considering how $\top$ is always true, we may even provide no resources at all) and combine it with what is already available, then it will be enough to make $\alpha$ true''. In short, in an alethic reading of modal logic, necessity of $\alpha$ is understood as $\alpha$ being true regardless of allocated resources.

Normal modal logic \SFour has classical operators as well as an intensional operator $\Box$. We occasionally use the symbol $\Diamond$ to refer to the De Morgan dual of $\Box$, defined as $\Diamond \alpha \equiv \neg \Box \neg \alpha$. We present the nowadays standard Hilbert-style axiomatic system for \SFour:

\begin{definition}[Axiomatic System \SFour]\label{Def:S4}
Add to an axiomatisation of classical propositional logic with the rule of \textit{modus ponens} the following axiom schemes and deduction rule:
\begin{enumerate}
\begin{multicols}{2}
\item [K.] $\Box (\phi \rightarrow \psi)\rightarrow (\Box\phi \rightarrow \Box \psi)$.
\item[T.] $\Box \phi \rightarrow \phi$.
\item[4.] $\Box \phi \rightarrow \Box \Box \phi$.
\item[R1.] \AxiomC{$\alpha$}
\RightLabel{[Nec]}
\UnaryInfC{$\Box\alpha$}

\DisplayProof

\end{multicols}

\end{enumerate}

\end{definition}

Rule R1 is referred as necessitation. System \SFour is sound and complete with respect to Kripke relational semantics based on both pre-orders (i.e. reflexive and transitive binary relations) and partial orders (i.e. reflexive, transitive, and antisymmetric binary relations). Intuitively, this happens because antisymmetry is not definable in modal logic (while it is in hybrid logic \cite{Blackburn_Rijke_Venema_2001}). In any case, partial orders are a particular class of pre-orders. We relay the standard relational model for \SFour:

\begin{definition}[S4 Kripke Model] Let $(W, R)$ be a relational frame based on a pre-order. Thus, the accessibility relation $R:W\longrightarrow W$ is reflexive and transitive. The forcing relation ``$\Vdash$'' behaves as usual for classical operators, with the added clause:

    \begin{description}
        \item $w \Vdash \Box\alpha$ iff $\forall w' \in W$ such that $wRw'$ $w'\Vdash \alpha$.

    \end{description}

If $w \Vdash \alpha$ for all $w$ in the model, we write $\vDash_{S4}\alpha$.
    
\end{definition}

As noted in \cite{galmiche_expressivity_2006}, we can obtain a pre-order from the ternary resource semantics for \BBI. 

\begin{definition}[Induced S4-Frame]\label{induced_frame}

Let $\mathcal{F}$ be a relational resource frame. Let the binary relationship relationship ``$\preceq$'' be defined as such $r_1, r_2 \in M$, $r_1\preceq r_2$ if there exists $r'\in M$ such that $r',r_1 \triangleright r_2$. Then, we define $\mathcal{U}(\mathcal{F})$ as the induced \SFour-frame obtained by collapsing the ternary relation $\triangleright$ into the binary $\preceq$ on the resource relational \BBI-frame $\mathcal{F}$. 
\end{definition}

This notion remains useful even if we add properties to the relational resource frame $\mathcal{F}$, as this will in turn define a particular sub-class of induced \SFour-frames. Validity for a formula (in a class of frames) is shared by the two frameworks: a formula holds in a frame iff it holds for each evaluation point in each model belonging to the particular class of frames.

\section{Embedding \SFour into \BBI}
\label{Sec:embedding}

We adopt the same translation as used by
Galmiche and Larchey-Wendling~\cite{galmiche_expressivity_2006}:

\begin{definition}[Translation from \SFour to \BBI]\label{Def:emb}
The translation $t(-)$ from the language of \SFour into the language of \BBI is:
\begin{description}
    \item[1.] $t(\Box \alpha) =_{Def}\top \magicwand t(\alpha) $.
    \item[2.] $t(\neg \alpha) =_{Def}\neg t(\alpha) $.
    \item[3.] $t(\alpha \circ \beta)=_{Def}t(\alpha) \circ t(\beta) $, for $\circ \in \{\land, \lor, \rightarrow\}$.
    \item[4.] $t(C)=_{Def}C$, for $C \in \{\top, \bot\}$ or $C \equiv p$ for $p$ a propositional atom.
\end{description}
\end{definition}

\begin{theorem}[Embedding is Sound]
\label{Theorem:sound}
If $\vdash_{\SFour}\alpha$ then $\vdash_{\SFour}t(\alpha)$.
\end{theorem}
\begin{proof}
Straightforward induction on the length of formulae.
\end{proof} 

The converse of this theorem is the \emph{faithfulness} of the translation.
Instead of employing the semantical approach of Galmiche and Larchey-Wendling~\cite{galmiche_expressivity_2006} we take inspiration from a methodology of G\"odel, recorded only in his 1940-1942 unpublished notebooks “\textit{Resultate Grundlagen}'' \cite{Godel_Grundlagen}, there used to prove the faithfulness of the embedding of intuitionistic propositional logic into \SFour.

The beating heart of the proof is a \emph{reverse} translation $\rev{-}$, from \BBI to \SFour.
This translation $\rev{-}$ must be proven sound but is certainly not faithful; if it were, the undecidable \BBI~\cite{5570897} could be decided by translation into the decidable \SFour. Instead we show that it is `cancelling', providing a sort of left inverse to the embedding $t(-)$.

\begin{proposition}[Reverse Translation is Sound]
\label{Prop:rev_sound}
For every formula $\alpha$ of \BBI, if $\vdash_{\BBI} \alpha$ then $\vdash_{\SFour} [\alpha]'$.
\end{proposition}

\begin{proposition}[Reverse Translation is Cancelling]
\label{Prop:cancel}
For every formula $\alpha$ in \SFour, $\vdash_{\SFour} \alpha \leftrightarrow \rev{t(\alpha)}$.
\end{proposition}

These propositions, once proved, will imply our main result:

\begin{theorem}[Embedding is Faithful] \label{Theor:Faithfulness}
 If $\vdash_{\BBI}t(\alpha)$, then $\vdash_{\SFour}\alpha$.
\end{theorem}
\begin{proof}If $\vdash_{\BBI}t(\alpha)$ then, by
Proposition~\ref{Prop:rev_sound}, $\vdash_{\SFour}\rev{t(\alpha)}$. But by Proposition~\ref{Prop:cancel}, we have $\vdash_{\SFour}\rev{t(\alpha)}\leftrightarrow\alpha$, hence the result.
\end{proof}

So far things have proceeded \`a la G\"odel, but as our logics of interest differ, we require a novel reverse translation: 

\begin{definition}[Reverse Translation]\label{Def:Reverse_Translation}
The translation $\rev{-}$ from the language of \BBI into the language of \SFour is:
\begin{description}
    \item[1.] $[\alpha \magicwand \beta]' =_{Def}[\alpha]' \rightarrow \Box[\beta]'\lor [\beta]'$.
    \item[2.]$[\alpha * \beta]' =_{Def}[\alpha]' \land [\beta]'$.
    \item[3.]$[I]'=_{Def}\top$. 
    \item[4.]$[\neg \alpha]'=_{Def}\neg [\alpha]'$ 
    \item[5.]$[\alpha \circ \beta]'=_{Def}[\alpha]'\circ[\beta]'$ for $\circ \in \{\land, \lor, \rightarrow\}$.
    \item[6.]For propositional variables $p$, $[p]'=_{Def}p$ 
\end{description}
    
\end{definition}

\begin{proof}[Proof of Proposition~\ref{Prop:rev_sound}]
    By induction on the derivation in \BBI. First, axioms of \BBI translate to formulae provable in \SFour. Indeed, as axioms do not contain $\magicwand$ these axioms translate into the language of propositional logic, and are clearly theorems of classical propositional logic:

  \begin{description}

   \item[1.] $[\alpha \rightarrow I*\alpha]' \;=\; [\alpha]' \rightarrow \top\land[\alpha]'$.

    \item[2.] $[I*\alpha \rightarrow \alpha]' \;=\; \top \land [\alpha]' \rightarrow [\alpha]'$.

    \item[3.] $[\alpha*\beta\rightarrow\beta*\alpha]' \;=\; [\alpha]' \land [\beta]'\rightarrow[\beta]' \land [\alpha]'$.

    \item[4.] $[\alpha*(\beta*\gamma) \rightarrow (\alpha *\beta) * \gamma)]' \;=\; [\alpha]' \land ([\beta]' \land [\gamma]') \rightarrow ([\alpha]' \land [\beta]') \land [\gamma]'$.

  \end{description}

    As the induction step, we see that the translations of inference rules of \BBI preserve derivability in \SFour:

\begin{description}

    \item[1.] Rule R1 (\textit{modus ponens}) is shared by \BBI and \SFour.

\item[2.] The translations of rules R2 and R3 correspond to asking that $(\alpha \land \beta)\rightarrow\gamma$ if and only if $\alpha \rightarrow (\beta \rightarrow (\Box \gamma \lor \gamma))$. In one direction:
\begin{description}
    \item[1.] $\vdash_{\SFour}\alpha \land \beta \rightarrow \gamma$.\textit{ Hyp.}
    \item[2.] $\vdash_{\SFour}(\alpha \land \beta \rightarrow \gamma) \rightarrow (\alpha \land \beta \rightarrow \gamma \lor \Box \gamma)$. \textit{Taut.}
    \item[3.] $\vdash_{\SFour}\alpha \land \beta \rightarrow \gamma \lor \Box \gamma$. \textit{[MP] on 1. and 2.}
    \item[4.] $\vdash_{\SFour}(\alpha \land \beta \rightarrow \gamma \lor \Box \gamma) \rightarrow (\alpha \rightarrow (\beta \rightarrow \Box \gamma \lor \gamma))$. \textit{Taut.}
    \item[5.] $\vdash_{\SFour}\alpha \rightarrow (\beta \rightarrow \Box \gamma \lor \gamma)$. \textit{[MP] on 3. and 4.}
    
\end{description}
In the other direction:
\begin{description}
    \item[1.] $\vdash_{\SFour}\alpha \rightarrow(\beta\rightarrow\Box \gamma \lor \gamma)$. \textit{ Hyp.}
    \item[2.] $\vdash_{\SFour}\Box \gamma \rightarrow \gamma$. \textit{ Modal Axiom Scheme \SFour.}
    \item[3.] $\vdash_{\SFour}(\alpha \rightarrow(\beta\rightarrow\gamma \lor \gamma).$ \textit{Classical Logic on 1. and 2.}
    \item[4.] $\vdash_{\SFour}\alpha \rightarrow(\beta\rightarrow \gamma).$ \textit{Classical Logic on 3.}
    \item[5.] $\vdash_{\SFour}((\alpha \rightarrow (\beta \rightarrow \gamma)) \rightarrow (\alpha \land \beta \rightarrow \gamma)$. \textit{Taut.}
    \item[6.] $\vdash_{\SFour}(\alpha \land \beta \rightarrow \gamma).$ \textit{[MP] on 4. and 5.}
\end{description}

    \item[3.] Translation of rule R4 states that from $(\alpha \rightarrow \gamma)$ and $(\beta \rightarrow \delta)$, we can conclude $( \alpha \land \beta)\rightarrow (\gamma \land \delta)$, which follows via classical propositional logic.

\end{description}

\end{proof}

\begin{proof}[Proof of Proposition~\ref{Prop:cancel}] We proceed by induction on the complexity of the formulae of \SFour. We briefly notice how each \SFour-axiom scheme $\alpha$ is \SFour-equivalent to its ``back-and-forth translation'' $\rev{t(\alpha)}$:

    \begin{description}
        \item[1.] $\vdash_{\SFour} (\Box(\alpha \rightarrow \beta) \rightarrow (\Box \alpha \rightarrow \Box \beta)) \leftrightarrow (\Box(\alpha \rightarrow \beta)\lor (\alpha \rightarrow \beta) \rightarrow (\Box \alpha \lor \alpha \rightarrow \Box \beta \lor \beta)).$ 
        \item[2.] $\vdash_{\SFour} (\Box \alpha \rightarrow \alpha) \leftrightarrow ((\Box \alpha \lor \alpha) \rightarrow \alpha).$   
        \item[3.] $\vdash_{\SFour}(\Box \alpha \rightarrow \Box \Box \alpha)\leftrightarrow (\Box(\Box \alpha\lor \alpha)\lor (\Box \alpha \lor \alpha)).$ 
        
    \end{description}
    
    As the basis of the induction, we notice how both $t(-)$ and $\rev{-}$ leave unaltered propositional atoms. Then, for the induction step, we proceed by structural induction on \SFour-formulae:

    \begin{description}

        \item[1.] The form of $\alpha$ is such that $\alpha \equiv \beta \land \gamma$. We want to show that $\vdash_{\SFour}(\beta\land\gamma) \leftrightarrow [t(\beta\land\gamma)]'$. Then, by definition of $t(-)$, $t(\beta \land \gamma)=t(\beta)\land t(\gamma)$. Thus, by definition of $\rev{-}$, $[t(\beta)\land t(\gamma)]'=[t(\beta)]'\land[t(\gamma)]'$. By induction hypothesis, we have $[t(\beta)]'\land[t(\gamma)]'=\beta \land \gamma$, and we conclude that $\vdash_{S4}(\beta \land \gamma) \leftrightarrow (\beta \land \gamma)$. The cases for remaining classical operators are analogous.
        
        \item[2.] The form of $\alpha$ is such that $\alpha \equiv \Box \beta$. We want to show that $\vdash_{\SFour}(\Box\beta) \leftrightarrow [t(\Box\beta)]'$. Then, by definition of $t(-)$, $t(\Box \beta)= \top \magicwand t(\beta)$. Thus, by definition of $\rev{-}$, we have that $[t(\Box \beta)]'= [\top \magicwand t(\beta)]'$. In turn, for definition of $\rev{-}$: $[\top \magicwand t(\beta)]' = [\top]'\rightarrow \Box [t(\beta)]'\lor [t(\beta)]'$. By induction hypothesis and definition of $\rev{-}$, then: $[\top]'\rightarrow \Box [t(\beta)]'\lor [t(\beta)]'=\top \rightarrow \Box \beta \lor \beta$. For classical propositional logic: $\vdash_{\SFour}(\top \rightarrow \Box \beta \lor \beta) \leftrightarrow (\Box \beta \lor \beta)$. Ultimately, for modal axiom $\Box \beta \rightarrow \beta$ of \SFour, we conclude that: $\vdash_{\SFour}\Box \beta \lor \beta \leftrightarrow \Box \beta$.
        
    \end{description}

\end{proof}

Having proved both lemmas, faithfulness of the translation $t(-)$ follows as stated in Theorem~\ref{Theor:Faithfulness}. Thus, we have proven the embedding.

\section{Derivability Under Assumptions}
\label{Sec:Assumptions}

While the calculi of Section~\ref{Sec:prelim} act on formulae alone,
it is usual to reason in the presence of assumptions,
and hence worthwhile to confirm the stability of our embedding with respect to derivability under assumptions.
We take our notion of \SFour-reasoning under assumptions from a yet to be published book by Negri~\cite{SaraNegriNewBook}:

\begin{definition}[Axiomatic \SFour with Reasoning Under Assumptions] A formula $\alpha$ is derivable from a multiset of assumptions $\Gamma$ in the system \textbf{HS4}, written $\Gamma \vdash_{\textbf{HS4}} \alpha$, if $\alpha$ is in $\Gamma$ or follows from \SFour axiom schemes, and applications of the rules of inference \textit{modus ponens} and \textit{necessitation}, where necessitation is applicable only to derivations without assumptions:
\begin{enumerate}
\begin{multicols}{2}

\item \AxiomC{$\alpha \in \Gamma $}
\UnaryInfC{$\Gamma \vdash \alpha$}

\DisplayProof

\vspace{0.5em}
\item \AxiomC{$\alpha \in Axiom $}
\UnaryInfC{$\Gamma \vdash \alpha$}

\DisplayProof

\item \AxiomC{$\Gamma \vdash \alpha$}
\AxiomC{$\Delta \vdash \alpha \rightarrow \beta$}

\BinaryInfC{$\Gamma, \Delta \vdash \beta$}

\DisplayProof

\vspace{0.5em}
\item \AxiomC{$\vdash \alpha$}
\UnaryInfC{$\Gamma \vdash \Box \alpha$}

\DisplayProof

\end{multicols}

\end{enumerate}
    
\end{definition}

Unlike most classical sequent calculi, we do not use multisets of formulae to the right of $\vdash$. We observe some properties of \textbf{HS4}, proved in more detail by Negri~\cite{SaraNegriNewBook}:

\begin{lemma}\label{lemma_HS4_a}
 If $\vdash_{\SFour} \alpha$, then $\vdash_{\textbf{HS4}} \alpha$.   
\end{lemma}
\begin{proof}[Proof Sketch:] By induction on the derivation in \textbf{S4} (ref. Definition~\ref{Def:S4}).
\end{proof}

We define the formula translation $\bigwedge$ from multisets of \SFour-formulae to \SFour-formulae as usual, by mapping ``$,$'' to $\land$ and the empty multiset to $\top$.

\begin{lemma}\label{lemma_HS4_b}
     $\Gamma \vdash_{\textbf{HS4}} \alpha$ iff $\vdash_{\SFour}\bigwedge\Gamma \rightarrow \alpha$. 
\end{lemma}
\begin{proof}[Proof Sketch:] ($\Rightarrow$): by induction on the structure of the \textbf{HS4} derivation. ($\Leftarrow$): by induction on the structure of the \SFour derivation, using the detachment theorem for $\rightarrow$.    
\end{proof}

\begin{observation}[Completeness for HS4] \label{HS4_Completeness}
Completeness follows from the previous results:
first note that $\Gamma \vDash_{\SFour}\alpha$ $\Leftrightarrow$ $\vDash_{\SFour} \bigwedge\Gamma\rightarrow \alpha$ $\Leftrightarrow$ $\vdash_{\SFour}\mathbf{T}(\Gamma) \rightarrow \alpha$.
Then by Lemma~\ref{lemma_HS4_a}, $\vdash_{\textbf{HS4}}\bigwedge\Gamma \rightarrow \alpha$, so by
Lemma~\ref{lemma_HS4_b}, $\Gamma\vdash_{\textbf{HS4}} \alpha$.
\end{observation}

The translation of Definition~\ref{Def:emb} from \SFour-formulae to \BBI-formulae can easily be extended to one from multisets of \SFour-formulae to multisets of \BBI-formulae by applying it pointwise, and indeed this suffices for the major result of this Section.
However it is worth noting that \emph{structural} proof theory for bunched logics typically organise assumptions in a richer way,
in the \emph{bunches} that give these logics their name~\cite{pym_semantics_2002}.
We will therefore also use this Section to develop the theory of derivability under bunches of assumptions.
In particular, bunches will allow us to state a version of the deduction theorem for $\magicwand$, which has been stated without proof for the positive fragment of \BBI~\cite{deduction_theorem} but not, to our knowledge, proved anywhere in the literature, so we will present a proof as a small contribution of this paper.

\begin{definition}\emph{Bunches} are generated by the grammar:
\[
  \Gamma \;::=\; \alpha \quad|\quad \varnothing_{Add} \quad|\quad \varnothing_{Mult} \quad|\quad (\Gamma ;\Gamma)
  \quad|\quad (\Gamma , \Gamma).
\]
where $\alpha$ ranges over \BBI-formulas.
We write $\Gamma[\Delta]$, and refer to $\Delta$ as a sub-bunch of $\Gamma$, for a bunch $\Gamma$ in which $\Delta$ appears as a sub-tree.
We then work modulo an equivalence $\equiv$ on bunches generated by the following equations:
\begin{itemize}
    \item[1.] Commutative monoid equations for $\varnothing_{Add}$ and ``$;$''.
    \item[2.] Commutative monoid equations for $\varnothing_{Mult}$ and ``$,$''.
    \item[3.] Congruence: if $\Delta\equiv\Delta'$, then $\Gamma[\Delta] \equiv \Gamma[\Delta']$.
\end{itemize}
\end{definition}

Bunches are therefore trees with formulae as leaves, and internal nodes labelled ``$,$'' or ``$;$''.
We may consider multisets of assumptions as the particular case of bunches in which only formulae and ``$;$'' appear,
or $\varnothing_{Add}$ to indicate the empty multiset.

We present two axiomatic systems for \BBI with assumptions: \textbf{HBBI}, which uses only multisets,
and \textbf{HBBI*}, which uses bunches:

\begin{definition}[Axiomatic \BBI with Reasoning Under Multisets of Assumptions]
     A formula $\alpha$ is derivable from a multiset of assumptions $\Gamma$ in the system \textbf{HBBI}, written $\Gamma \vdash_{\textbf{HBBI}} \alpha$, according to the following rules, where $Axiom$ contains the \BBI axiom schemes.
\begin{enumerate}
\begin{multicols}{2}

\item \AxiomC{$\alpha \in \Gamma $}
\UnaryInfC{$\Gamma \vdash \alpha$}

\DisplayProof

\vspace{0.5em}
\item \AxiomC{$\alpha \in Axiom $}
\UnaryInfC{$\Gamma \vdash \alpha$}

\DisplayProof

\vspace{0.5em}
\item \AxiomC{$\Gamma \vdash \alpha $}
\AxiomC{$\Delta \vdash \alpha \rightarrow \beta$}

\BinaryInfC{$\Gamma; \Delta \vdash \beta$}

\DisplayProof

\item \AxiomC{$ \vdash \alpha \rightarrow \gamma $}
\AxiomC{$ \vdash \beta \rightarrow \delta$}

\BinaryInfC{$\Gamma \vdash (\alpha * \beta) \rightarrow (\gamma * \delta)$}

\DisplayProof

\vspace{0.5em}
\item \AxiomC{$ \Gamma \vdash \alpha \rightarrow (\beta \magicwand \gamma)$}
\UnaryInfC{$\Gamma \vdash (\alpha *\beta) \rightarrow \gamma$}

\DisplayProof

\vspace{0.5em}
\item \AxiomC{$ \Gamma\vdash (\alpha *\beta) \rightarrow \gamma$}
\UnaryInfC{$\Gamma \vdash \alpha \rightarrow (\beta \magicwand \gamma)$}

\DisplayProof

\end{multicols}

\end{enumerate}

\end{definition}

\begin{definition}[Axiomatic \BBI with Reasoning Under Bunches of Assumptions]
     A formula $\alpha$ is derivable from a bunch of assumptions $\Gamma$ in the system \textbf{HBBI*}, written $\Gamma \vdash_{\textbf{HBBI}} \alpha$, according to the following rules, where $Axiom$ contains the \BBI axiom schemes.
\begin{enumerate}
\begin{multicols}{2}

\item \AxiomC{}
\UnaryInfC{$\alpha \vdash \alpha$}

\DisplayProof

\vspace{0.5em}
\item \AxiomC{$\Gamma[\Gamma'] \vdash \alpha$}
\UnaryInfC{$\Gamma [\Gamma'; \Gamma''] \vdash \alpha$}

\DisplayProof

\vspace{0.5em}
\item \AxiomC{$\Gamma[\Gamma'; \Gamma' ]\vdash \alpha$}
\UnaryInfC{$\Gamma[\Gamma']\vdash \alpha$}

\DisplayProof

\vspace{0.5em}
\item \AxiomC{$\alpha \in Axiom $}
\UnaryInfC{$\Gamma \vdash \alpha$}

\DisplayProof

\vspace{0.5em}
\item \AxiomC{$\Gamma \vdash \alpha $}
\AxiomC{$\Delta \vdash \alpha \rightarrow \beta$}

\BinaryInfC{$\Gamma; \Delta \vdash \beta$}

\DisplayProof

\item \AxiomC{$ \vdash \alpha \rightarrow \gamma $}
\AxiomC{$ \vdash \beta \rightarrow \delta$}

\BinaryInfC{$\Gamma \vdash (\alpha * \beta) \rightarrow (\gamma * \delta)$}

\DisplayProof

\vspace{0.5em}
\item \AxiomC{$ \Gamma \vdash \alpha \rightarrow (\beta \magicwand \gamma)$}
\UnaryInfC{$\Gamma \vdash (\alpha *\beta) \rightarrow \gamma$}

\DisplayProof

\vspace{0.5em}
\item \AxiomC{$\Gamma \vdash (\alpha *\beta) \rightarrow \gamma$}
\UnaryInfC{$\Gamma \vdash \alpha \rightarrow (\beta \magicwand \gamma)$}

\DisplayProof

\vspace{0.5em}
\item \AxiomC{$\Gamma \vdash \alpha $}
\AxiomC{$\Delta \vdash \beta$}

\BinaryInfC{$\Gamma, \Delta \vdash \alpha*\beta$}
\DisplayProof

\end{multicols}

\end{enumerate}

\end{definition}

\begin{lemma}\label{lemma_HBBI_a} If $\vdash_{\BBI}\alpha$, then $\vdash_{\textbf{HBBI}/\textbf{HBBI*}} \alpha$.
    
\end{lemma}

\begin{proof}[Proof Sketch:] We proceed by induction on the structure of the derivation of $\alpha$ in \BBI. For the full detail of the proof, we refer to the Appendix.    
\end{proof}

\begin{definition}[Formula Translation for Bunches]\label{Formula_Translation}
The translation $\mathcal{T}(-)$ from bunches to \BBI formulae is:
\begin{description}
    \item[1.] $\mathcal{T}(\alpha) =_{Def} \alpha $.
    \item[2.] $\mathcal{T}(\varnothing_{Add}) =_{Def} \top $.
    \item[3.] $\mathcal{T}(\varnothing_{Mult}) =_{Def} I $.
    \item[4.] $\mathcal{T}(\Gamma;\Gamma') =_{Def} \mathcal{T}(\Gamma)\land\mathcal{T}(\Gamma') $.
    \item[5.] $\mathcal{T}(\Gamma,\Gamma') =_{Def} \mathcal{T}(\Gamma)*\mathcal{T}(\Gamma') $.
\end{description}
This restricts to a translation from multisets in the evident way.
\end{definition}

\begin{lemma}\label{HBBI_Completeness_b1} If $\Gamma \vdash_{\textbf{HBBI}/\textbf{HBBI*}} \alpha$, then $\vdash_{\BBI} \mathcal{T}(\Gamma) \rightarrow \alpha$. 
\end{lemma}

\begin{proof}[Proof Sketch:]

By induction on the structure of the \textbf{HBBI}/\textbf{HBBI*} derivation. For full details we refer to the Appendix.
\end{proof}

The converse of this lemma, and hence completeness of \textbf{HBBI}/\textbf{HBBI*}, 
requires deduction and detachment theorems.

\begin{theorem}[Deduction Theorems]\label{Deduction_theorem} \
\begin{enumerate}
    \item[(a)] If $\alpha; \Gamma \vdash_{\textbf{HBBI}} \beta$, then $\Gamma \vdash_{\textbf{HBBI}} \alpha \rightarrow \beta$.
    \item[(b)] If $\alpha; \Gamma \vdash_{\textbf{HBBI*}} \beta$, then $\Gamma \vdash_{\textbf{HBBI*}} \alpha \rightarrow \beta$.
    \item[(c)] If $\alpha, \Gamma \vdash_{\textbf{HBBI*}} \beta$, then $\Gamma \vdash_{\textbf{HBBI*}} \alpha \magicwand \beta$.
\end{enumerate}
    
\end{theorem}
\begin{proof}[Proof Sketch:] By induction on the structures of the \textbf{HBBI}/\textbf{HBBI*} derivations.
We use Lemma \ref{HBBI_Completeness_b1} for (a) and (b), and Lemma \ref{lemma_HBBI_a} for (c). For full proofs, see the Appendix.
\end{proof}

\begin{theorem}[Detachment Theorems] \label{detachment} \
\begin{enumerate}
    \item[(a)] If $\Gamma \vdash_{\textbf{HBBI}}\alpha \rightarrow \beta$, then $\alpha; \Gamma \vdash_{\textbf{HBBI}} \beta$.

    \item[(b)] If $\Gamma \vdash_{\textbf{HBBI*}}\alpha \rightarrow \beta$, then $\alpha; \Gamma \vdash_{\textbf{HBBI*}} \beta$.
    
    \item[(c)] If $\Gamma\vdash_{\textbf{HBBI*}}\alpha \magicwand \beta$, then $\Gamma, \alpha \vdash_{\textbf{HBBI*}} \beta$.\label{detachment_theorem_magicwand}    

\end{enumerate}

\end{theorem}

\begin{proof}[Proof Sketch:] For (a) and (b), simply apply \textit{modus ponens} with $\Gamma \vdash\alpha \rightarrow \beta$ and $\alpha \vdash \alpha$
For (c) we show that a version of \textit{modus ponens} for $\magicwand$ is admissible in \textbf{HBBI*}; for full details we refer to the Appendix.    
\end{proof}

\begin{lemma}\label{HBBI_Completeness_b2}
 If $\vdash_{\BBI} \mathcal{T} (\Gamma) \rightarrow \alpha$, then $\Gamma \vdash_{\textbf{HBBI/HBBI*}}\alpha$.   
\end{lemma}

\begin{proof}[Proof Sketch:] $\vdash_{\textbf{HBBI/HBBI*}}\mathcal{T}(\Gamma)\rightarrow \alpha$ by Lemma \ref{lemma_HBBI_a} , so by detachment on $\rightarrow$ we have $\mathcal{T}(\Gamma)\vdash_{\textbf{HBBI/HBBI*}}\alpha$.
In the multiset case this can easily be shown to be equivalent to $\Gamma\vdash_{\textbf{HBBI}}\alpha$.
For the \textbf{HBBI*} result we require an induction on the \BBI derivation, and the deduction and detachment theorems for $\magicwand$. For full details, we refer to the Appendix.
\end{proof}

\begin{observation}
\sloppy Completeness for \textbf{HBBI/HBBI*} with respect to partial non-deterministic monoid semantics follows analogously to \textbf{HS4} in Observation~\ref{HS4_Completeness}.  
\end{observation}

We now set the bunched system aside to prove the stability of our embedding on reasoning with assumptions.
We write $t(\Gamma)$ for the pointwise application of $t(-)$ to formulae, also replacing ``$,$'' with ``$;$'' in accordance with the usual notation for structural connectives in \BBI.

\begin{theorem}[Embedding Holds Under Assumptions] $\Gamma \vdash_{\textbf{HS4}}\alpha$ iff $t(\Gamma)\vdash_{\textbf{HBBI}}t(\alpha)$.
\end{theorem}

\begin{proof}
\begin{description}

    \item[$\Rightarrow$ (Soundness):] $ \vdash_{\SFour} \bigwedge \Gamma \rightarrow \alpha$ by Lemma~\ref{lemma_HS4_b}.
By soundness of the translation $t(-)$ and its definition, we get $\vdash_{\BBI}t(\bigwedge \Gamma \rightarrow \alpha)\equiv t(\bigwedge \Gamma)\rightarrow t(\alpha)$. By Lemma~\ref{lemma_HBBI_a}, $\vdash_{\textbf{HBBI}} t(\bigwedge \Gamma)\rightarrow t(\alpha)$, and by detachment on ``$\rightarrow$'' in \textbf{HBBI}, we get $t(\bigwedge \Gamma) \vdash_{\textbf{HBBI}}t(\alpha)$. But the translation of the multiset containing the conjunction of the formulae of the original $\Gamma$ is nothing but the conjunction of the translations of each of the formulae in $\Gamma$. This is equivalent to the multiset $t(\Gamma)$, containing the translations of each formula of the original $\Gamma$.

\item[$\Leftarrow$ (Faithfulness):] By Lemma~\ref{HBBI_Completeness_b1} we have $\vdash_{\BBI}\mathcal{T} (t(\Gamma))\rightarrow t(\alpha)$. But $\mathcal{T} (t(\Gamma)) \equiv t(\bigwedge \Gamma)$, so we have $\vdash_{\BBI}t(\bigwedge\Gamma)\rightarrow t(\alpha) \equiv t(\bigwedge \Gamma \rightarrow \alpha)$ by the definition of $t(-)$. By faithfulness of the translation we get $\vdash_{\SFour}\bigwedge \Gamma \rightarrow \alpha$. By Lemma~\ref{lemma_HS4_a} we obtain $\vdash_{\textbf{HS4}}\bigwedge \Gamma \rightarrow \alpha$, and by detachment $\bigwedge \Gamma\vdash_{\textbf{HS4}} \alpha$, and is straightforward to conclude that
$\Gamma\vdash_{\textbf{HS4}} \alpha$.    
\end{description}    
\end{proof}

\section{Generalised Embedding}
\label{Sec:general}

Unlike the semantic proof of the embedding by Galmiche and Larchey-Wendling~\cite{galmiche_expressivity_2006},
our proof does not depend on particular semantic features of \SFour. Thus, our result can be generalised to arbitrary
axiomatic extensions of \SFour. Further, it can be generalised to a class of axiomatic extensions of \BBI that includes,
to our knowledge, all such extensions that have been studied in the literature.

\begin{proposition}[Generalised Embedding] Let $\alpha$ be a modal logic formula and $\beta$ a \BBI-formula such that $\vdash_{\SFour + \alpha}[\beta]'$. Then: $\Gamma\vdash_{S4+Ax+[\beta]'} \alpha$ iff $t(\Gamma)\vdash_{\BBI+t(Ax)+\beta} t(\alpha)$.     
\end{proposition}

While the modal formula $\alpha$ is unconstrained, the \BBI-formula $\beta$ is not.
It is therefore an empirical observation that the condition on $\beta$ does not seem restrictive in practice:

\begin{observation} Reverse translation $\rev{-}$ maps the following axioms into \SFour-provable formulae. The first three axioms define \BBI over models with, respectively, an indivisible unit (\textbf{BBI.iu})~\cite{separation_theories}, the divisibility property (\textbf{BBI.d}), and totality (\textbf{BBI.Tot}) \cite{total_monoid_signature}.
The fourth axiom is a *-elimination rule suggested by Cao et al.~\cite{Bringing_order}.

\begin{description}
    \item[BBI.iu.] $[(I \land (\alpha * \beta))\rightarrow \alpha]'\mapsto((\top \land ([\alpha]' \land [\beta]'))\rightarrow [\alpha]')$.
    \item[BBI.d.] $[(\neg I \rightarrow (\neg I * \neg I))]'\mapsto (\neg \top \rightarrow (\neg \top \land \neg \top))$.
     \item[BBI.Tot] $[I\land((\alpha*\beta)\magicwand \bot) \rightarrow ((\alpha \magicwand \bot)\lor (\beta \magicwand \bot))]'\mapsto (\top \land (\alpha \land \beta)\rightarrow \Box \bot \lor \bot) \rightarrow ((\alpha \rightarrow \Box \bot \lor \bot)\lor (\beta \rightarrow \Box \bot \lor \bot))$.
     \item[*-elimination.] $[(\alpha * \beta) \rightarrow \alpha]'\mapsto ([\alpha]' \land
    [\beta]' \rightarrow [\alpha]')$.

\end{description}

\end{observation}
    
We now explain how the generalisation of the embedding to axiomatic extensions of \SFour relates to semantics, via the derivation of \SFour-frames from \BBI-frames of Definition~\ref{induced_frame}.

\begin{theorem}
For $\phi$ a modal logic formula, $\phi$ holds in $\mathcal{U}(\mathcal{F})$ iff $t(\phi)$ holds in $\mathcal{F}$.  
\label{theor:induced_frames}
\end{theorem}

\begin{proof}
By induction on the structure of $\phi$.

\begin{description}
    \item[1.] $\phi \equiv \alpha \land \beta$ (the cases for the remaining classical operators are similar): if $\phi$ holds in all models of a frame $\mathcal{U}(\mathcal{F})$, then $\forall w, \forall M \in \mathcal{U}(\mathcal{F})$, $w\Vdash \alpha$ and $w\Vdash \beta$. Thus, we apply the induction hypothesis, and obtain that $\forall w, \forall M \in \mathcal{F}$ it is the case that $w\Vdash t(\alpha)$ and $w\Vdash t(\beta)$. Conversely, if $t(\alpha \land \beta)\equiv t(\alpha)\land t(\beta)$ holds in $\mathcal{F}$, then $\forall w, \forall M \in \mathcal{F}$, $w\Vdash t(\alpha)$ and $w\Vdash t(\beta)$. We then apply the induction hypothesis.

\item[2.] $\phi \equiv \Box \alpha$: If $\phi$ holds in all models of a frame $\mathcal{U}(\mathcal{F})$ (that is, $\forall w, \forall M \in \mathcal{U}(\mathcal{F})$), then also $\forall w'$ such that $wRw'$, $\forall M \in \mathcal{F}$, $w'\Vdash \alpha$. We apply the induction hypothesis and obtain that $\forall w', \forall M \in \mathcal{U}(\mathcal{F})$, is the case that $w'\Vdash t(\alpha)$. Thus, for a generic world $w$ such that $\exists w^*$ and $w^*,w \triangleright w'$, we have both $w^*\Vdash \top$ and $w'\Vdash t(\alpha)$. Hence $w\Vdash \top \magicwand t(\alpha)$ holds in $\mathcal{F}$.

Conversely, say $t(\Box \alpha)\equiv \top \magicwand t(\alpha)$ holds in $\mathcal{F}$. Then $t(\alpha)$ holds in all extensions of all worlds: $\forall w \forall w' \exists w^*$ such that $w^*, w\triangleright w'$ and $w^*\Vdash \top$, $w'\Vdash t(\alpha)$. Ny the identity condition $e, r1 \triangleright r1$ on the frame $\mathcal{F}$, and that $e\Vdash \top$, we notice that this holds for all worlds in all the models of the frame, as all worlds are extensions of themselves via the empty resource $e$. We then apply the induction hypothesis and obtain that $\forall w, \forall M \in \mathcal{U}(\mathcal{F})$, $\Vdash \alpha$. By the reflexivity of the induced pre-order $\preceq$, we have that each world $w$ in all the models of the frame $\mathcal{U}(\mathcal{F})$ sees itself. Hence $\Vdash \Box \alpha$ holds in the frame $\mathcal{U}(\mathcal{F})$.
\end{description}    
\end{proof}

\begin{corollary}
Suppose we have some class of \SFour-frames $\mathcal{K}$ and an \SFour-formula $\phi$ such that, for all \SFour-frames $\mathcal{G}$, $\mathcal{G} \in \mathcal{K}$ iff $\mathcal{F}\vDash \phi$. Then, for any \BBI-frame $\mathcal{F}$, $\mathcal{U}(\mathcal{F})\in \mathcal{K}$ iff $\mathcal{F}\vDash t(\phi)$.

\end{corollary}
\begin{proof}
Immediate from Theorem~\ref{theor:induced_frames}.    
\end{proof}

\begin{table}[t]
\centering
\scriptsize
\setlength{\tabcolsep}{12pt}   % extra spacing between columns
\renewcommand{\arraystretch}{1.25}

\begin{tabular}{p{3.2cm} p{3.2cm} p{6.2cm}}
\toprule
\textbf{S4 Axiomatic Extension} &
\textbf{BBI Axiomatic Extension} &
\textbf{BBI-Frame Condition} \\
\midrule

\textbf{S4.2} = \SFour + $(\Diamond\Box\alpha \rightarrow \Box\Diamond\alpha)$ &
\BBI + $(\top \septraction (\top \magicwand \alpha)) \rightarrow (\top \magicwand (\top \septraction \alpha))$ & \textbf{Confluence}:
$\forall a,b,c.\; (\exists m',m''.\; a,m' \triangleright b \wedge a,m'' \triangleright c)
\Rightarrow (\exists d,k',k''.\; b,k' \triangleright d \wedge c,k'' \triangleright d)$ \\
\midrule

\textbf{S4.3} = \SFour + $\Box(\Box\alpha \rightarrow \beta)\lor \Box(\Box\beta \rightarrow \alpha)$ &
\BBI + $(\top \magicwand ((\top \magicwand \alpha)\rightarrow \beta)) \lor (\top \magicwand ((\top \magicwand \beta)\rightarrow \alpha))$ & \textbf{Linearity}:
$\forall a,b.\; \exists m'.\; (a,m' \triangleright b) \lor (b,m' \triangleright a)$ \\
\midrule

\textbf{S5} = \SFour + $(\alpha \rightarrow \Box\Diamond\alpha)$ &
\BBI + $(\alpha \rightarrow (\top \magicwand (\top \septraction \alpha)))$ & \textbf{Transitivity}:
$\forall a,b,c.\; (\exists m',m''.\; a,m' \triangleright b \wedge b,m'' \triangleright c)
\Rightarrow (\exists m'''.\; a,m''' \triangleright c)$

\textbf{Symmetry}: 
$\forall a,b.\; \exists m',m''.\; (a,m' \triangleright b) \Leftrightarrow (b,m'' \triangleright a)$ \\
\bottomrule
\end{tabular}

%\vspace{0.6em} 

\caption{Some \BBI-extensions obtained via extensions of \SFour.}
\label{tab:bbi-extensions}
\end{table}

We can use this fact to provide a number of \BBI-extensions, sound and complete with regard to a particular class of \BBI-frames, for example those of Table~\ref{tab:bbi-extensions}. An intuitive reading of some of the added conditions on the resource frames can be proposed. For example, the linearity condition states that for any pair of resources, we can always extend one to obtain the other. For example, a wallet with 5 dollars in it can be extended to one with 20 dollars in it, or a heap assigning 5 addresses can be extended to one assigning 20 addresses, if we remain agnostic about their values. Confluence is the weaker condition that any two pairs of resources can be each extended to the same resource. A wallet with 5 dollars in it and one with 20 euros in it are not comparable, in the absence of a money changer, but they can each be extended to a wallet containing both currencies.

\section{Language Extensions}
\label{Sec:language_ext}

Beyond the axiomatic extensions of the previous Section, a number of language extensions have been proposed for \BBI to increase its expressivity. Such extensions in general are not constrained enough to support the proof of a general theorem about embeddings, but in this Section we investigate two specific approaches to extending the language of \BBI: first, hybrid extensions of \BBI~\cite{separation_theories}, and second, Classical Bunched Implication \textbf{CBI} \cite{CBI_Axioms} and the intermediate class of \textbf{BiBBI} logics \cite{intermediateBBI}, which extend the language with new connectives including a multiplicative analogue of disjunction. We observe that our techniques deal successfully with these extensions.

Hybrid logic is an extension of modal logic in which a countable set of special atomic formulae, called ``nominals'', is used to refer to evaluation points of the model \cite{Hybrid_Logic_Pure_extensions}. Moreover, a satisfaction one-place operator indexed to nominals ``$@_n \alpha$'' is used to convey the statement that $\alpha$ is true at evaluation point $n$. Thus, in basic hybrid language $\mathcal{H}(@)$, we also have formulae of shape $n, m,...$ for $n, m, ... \in Nom$, and $@_n \alpha$ for $\alpha \in Form$, $n\in Nom$. When a formula has no occurrences of propositional or formula variables (and thus, only nominals occur), it is called ``pure''. The axiomatic hybrid system \textbf{K}$_{\mathcal{H}(@)}$ is sound and complete with regard of the class of all frames. This system is a conservative extension of modal logic \textbf{K}, and has also been proven complete with regard to arbitrary extensions with pure axioms \cite{Hybrid_Logic_Pure_extensions}. It is thus straightforward for us to obtain the (still) sound and complete axiomatic system \SFour$_{\mathcal{H}(@)}$ by adding pure axioms expressing reflexivity and transitivity.

A number of interesting properties of certain classes of \BBI-models are not definable through \BBI-formulae, but are definable in its hybrid logic extension~\cite{separation_theories}. These properties are referred as ``separation properties'', and the set of axioms holding iff the model has those properties, ``separation theories''. We relate this powerful system to the system \SFour$_{\mathcal{H}(@)}$:

\begin{theorem}[Embedding of \SFour$_{\mathcal{H}(@)}$ into \textbf{HyBBI}] \label{Hybrid_Embedding} Let $\alpha$ be a hybrid logic formula and $\beta$ a \textbf{HyBBI}-formula such that $\vdash_{\SFour_{\mathcal{H}(@)} + \alpha}[\beta]'$. Then: $\vdash_{\SFour_{\mathcal{H}(@)+Ax+[\beta]'}} \alpha$ iff $\vdash_{\textbf{HyBBI}+t(Ax)+\beta} t(\alpha)$.

\end{theorem}

\begin{proof}[Proof Sketch:]
The translation employed is an extension of $t(-)$ that leaves the hybrid operators unaltered, as does the reverse $\rev{-}$ translation used in the proof. We prove soundness of the embedding by the usual proof by induction (proof~\ref{Theor:Soundness_HyBBI}), while faithfulness is established through G\"odel's reverse translation strategy (proof~\ref{Theor:Faithfulnes_1_HyBBI} and~\ref{Theor:Faithfulnes_2_HyBBI}).
    
\end{proof}

Again, the hybrid formula $\alpha$ is unconstrained, while \textbf{HyBBI}-formula $\beta$ is not. It is therefore an empirical observation that the condition on $\beta$ does not seem restrictive in practice:

\begin{observation}
All \textbf{HyBBI} axioms used to define separation properties (Table~\ref{Table:Separation_Theories}) contain only occurrences of nominals, $I$, $@_i$ or $*$. As the reverse translation easily shows, this means that all of these \textbf{HyBBI}-extensions still map into \SFour$_{\mathcal{H}(@)}$. An extension of the embedding on the extended hybrid language $\mathcal{H}(@, \downarrow)$, accounting for $\downarrow$ binder of nominals (and thus, able to express the separation property of ``splittability'') is straightforward. Moreover, the lack of $\magicwand$ in the signature scheme of such an extension suggests that the system still embeds into \SFour$_{\mathcal{H}(@, \downarrow)}$.  
\end{observation}

\begin{table}[t]
\centering
\scriptsize
\setlength{\tabcolsep}{12pt}   
\renewcommand{\arraystretch}{1.25}

\begin{tabular}{p{3.2cm} p{3.2cm} p{6.2cm}}
\toprule
\textbf{Name} &
\textbf{Semantic Property} &
\textbf{Separation Theory} \\
\midrule

Divisibility & $\forall w \cancel{\in} E, \exists w_1,  w_2. $  $w \in w_1 \circ w_2$& $\neg I \rightarrow (\neg I * \neg I)$\\
\midrule

Indivisible Units & $(w \circ w') \cap E \cancel{=}\varnothing \Rightarrow w \in E$
 & $(I \land (\alpha * \beta)) \rightarrow \alpha$ or $I \land (l_1 * l_2) \rightarrow l_1$ \\
\midrule

Functionality & $w,w'\in w_1 \circ w_2 \Rightarrow w = w'$
 & $@_l(j*k)\land @_{l'}(j*k) \rightarrow @_l l'$ \\
\midrule

Cancellativity & $(w \circ w_1) \cap (w \circ w_2) \cancel{=}\varnothing \Rightarrow w_1=w_2$ & $(l*j) \land (l*k) \rightarrow @_jk$\\
\midrule

Single Unit & $|E|=1$ & $@_{l_1}I \land@_{l_2}I \rightarrow @_{l_1}l_2$\\
\midrule

Disjointness & $w \circ w \cancel{=}\varnothing \Rightarrow w \in E$ & $l*l \rightarrow I \land l$\\

\midrule

Cross-split &  $(t \circ u) \cap (v \circ w) \cancel{=} \varnothing \Rightarrow \exists tv, tw,  uv, uw.$ $t \in tv \circ tw, u \in uv \circ uw, v \in tv \circ uv, w \in tw \circ uw$  & $(a*b)\land (c*d) \rightarrow @_a(\top *\downarrow ac.@a(\top * \downarrow ad.@_a(ac * ad)
\land @_b(\top *\downarrow bc.@b(\top *\downarrow bd.@_b(bc * bd)
\land @_c(ac * bc)\land @_d(ad * bd)))))$\\

\bottomrule
\end{tabular}

\caption{Separation theories \cite{separation_theories} definable in \BBI, \textbf{HyBBI}, or \textbf{HyBBI}($\downarrow$). We here use notation $s \circ t$ to mean $\{v|s,t \triangleright v\}$.}
\label{Table:Separation_Theories}
\end{table}

We finally consider the family of \textbf{BiBBI}-logics introduced by Brotherston and Villard~\cite{intermediateBBI}. Logics of the family are obtained by adding further multiplicative operators to \BBI: intuitionistic multiplicative disjunction $\lor^*$, its adjoint co-implication $\backslash^*$, and a multiplicative falsum constant $\bot^*$.
The logics of interest definable in this language include Classical BI (\textbf{CBI}) and Basic \textbf{BiBBI}, which will write as \textbf{BiBBI.b}. This logic \textbf{BiBBI.b} may be defined axiomatically as follows:

\begin{definition}[Axiomatic System for BiBBI.b] Basic \textbf{BiBBI} extends the proof system for \BBI with the following axiom schemes and inference rules:
    \begin{enumerate}
\begin{multicols}{2}
\item[R1.] \AxiomC{$(\alpha \rightarrow \gamma)$}
\AxiomC{$(\beta \rightarrow \delta)$}

\RightLabel{[Mon.]}
\BinaryInfC{$( \alpha \lor^* \beta)\rightarrow (\gamma \lor^* \delta)$}

\DisplayProof
\item[R2.] \AxiomC{$\alpha \rightarrow (\beta \lor^* \gamma)$}
\RightLabel{[Residuation 1]}
\UnaryInfC{$(\alpha \backslash^* \beta)\rightarrow \gamma$}

\vspace{0.5em}
\DisplayProof

\item[R3.] \AxiomC{$(\alpha \backslash^* \beta)\rightarrow \gamma$}
\RightLabel{[Residuation 2]}
\UnaryInfC{$\alpha \rightarrow (\beta \lor^* \gamma)$}

\DisplayProof

\item[Ax.] \vspace{0.5em}$\vdash (\alpha \lor^* \beta) \rightarrow (\beta \lor ^* \alpha)$.
\end{multicols}
\end{enumerate}

\end{definition}

The translation $t(-)$ of this paper also works as a parametric embedding from \SFour to the \textbf{BiBBI} logics described by Brotherston and Villard~\cite{intermediateBBI}, including \textbf{CBI}:

\begin{theorem}[Embedding of BiBBI-logics into \SFour]\label{Theor:Embedding_BiBBI} Let $\alpha$ be a modal logic formula and $\beta$ a \textbf{BiBBI}-formula such that $\vdash_{\SFour + \alpha}[\beta]'$. Then: $\vdash_{S4+Ax+[\beta]'} \alpha$ iff $\vdash_{\textbf{BiBBI.b}+t(Ax)+\beta} t(\alpha)$.     

\end{theorem}
\begin{proof}
    Soundness of the embedding is obtained through induction, and is straightforward. For proving the faithfulness, we extend $\rev{-}$ as follows:

\begin{description}

    \item[] $[\alpha \backslash^* \beta]'=_{Def} [\alpha]' \land \neg [\beta]'$.

    \item[] $[\alpha \lor^* \beta]'=_{Def} [\alpha]' \lor  [\beta]'$.

    \item[]  $[\bot^*]'=_{Def} \bot$.
    
\end{description}

 The translations of the \textbf{BiBBI.b} inference rules and axioms are admissible in classical logic (and moreover the translations of the axiomatic extensions considered by Brotherston and Villard~\cite{intermediateBBI} are provable in \SFour). Thus, soundness of $\rev{-}$ easily obtains. Faithfulness of the embedding obtains just as in Lemma~\ref{Theor:Faithfulness}, as the proof that $\rev{-}$ is cancelling is straightforward.

\end{proof}

\section{Conclusions: Related and Future Work}
\label{Sec:conc}

\textbf{\emph{Weaker base logics.}} While this work, following Galmiche and Larchey-Wendling~\cite{galmiche_expressivity_2006}, took \BBI and \emph{classical} \SFour as its basis, we conjecture a similar embedding between \textbf{BI} with an intuitionistic base, and \emph{intuitionistic} \SFour~\cite{Simpson94}.
If these embeddings were parametric, that would be interesting in light of work by Cao et al.~\cite{Bringing_order} which aims to present a unifying semantics for separation logic, and in particular looks at intermediate logics for the additive connectives, with properties such as weak excluded middle
(this should not be confused with the \textbf{BiBBI} logics of the previous Section, which are intermediate with respect to the multiplicative connectives).
Care must be taken if we reach below \BBI on how the relational semantics for $\to$ and $\magicwand$ are to be related; for example, much work on separation logic with an intuitionistic base defines the semantics of $\to$ by the resource extension notion of Definition~\ref{induced_frame}, but this is a choice that has logical consequences.

\textbf{\emph{Modal extensions of \BBI.}} Modalities have been added to \BBI and separation logic for a variety of purposes. For example Epistemic separation logic~\cite{courtault:hal-01259768} extends the \BBI syntax with a knowledge operator $K$, understood via an additional accessibility relation in the semantics obeying the rules of modal logic \textbf{S5}. The approach of this paper differs by deriving necessity within the language of \BBI.

\textbf{\emph{Reverse translations.}} While the methodology of G\"odel that inspired this paper has only recently been rediscovered~\cite{Godel_Grundlagen}, there is prior work translating from bunched logics to modal logics:
Brotherson et al.~\cite{Brotherston2010ClassicalBI,intermediateBBI} employ such translations in order to apply Sahlqvist theorems for modal logics.
Their work differs from ours because the modal logics used, called `modal similarity types', are technical tools and not of independent interest, and because they do not use their translation to establish an embedding in the other direction.
It is also worth noting that our results are attained in somewhat simpler fashion than those of G\"odel, for whom the reverse translation involved a transformation into a certain normal form, which required a propositional version of Barr's theorem.

\textbf{\emph{A lattice of bunched companions.}}
The structure of the relationships between intermediate and modal logics via translation has been a major theme of interpretational proof theory, with highlights such as the Blok-Esakia theorem \cite{chagrov_modal_1992}. A similar lattice of connections between bunched and modal logics could be investigated.
In particular, our observation that all previously studied axiomatic \BBI-extensions (\cite{separation_theories, total_monoid_signature, Bringing_order}) map into \SFour raises the question of establishing a least and greatest bunched companion of \SFour.

\textbf{\emph{Justification logics.}}
In the network of embeddings of Figure~\ref{fig:embeddings}, justification logic plays an interesting role, as the translation from \SFour to \textbf{LP} requires a sophisticated `realisation' argument that works via structural proof theory rather than mere axiomatisation.
We are interested in how a more direct link between \BBI and a justification logic might proceed.

\textbf{\emph{An epistemic interpretation.}}
We suggest a philosophical interpretation of our embedding, where we consider our resources as \emph{evidence}, so that a resource forces a formula if it provides supporting (but not necessarily conclusive) evidence for it.
The additive connectives of \BBI then relate to the combination and extension of evidence, while knowledge is a derived notion, via the derived $\Box$ picked out by the translation $t(-)$: combining any new evidence with that which we already have, a given formula remains supported. In short, knowledge is that which is \textit{indefeasible under new evidence}. A number of plausible logics for knowledge lie between \SFour and \textbf{S5}, such as Stalnaker's logic of knowledge and belief based on \textbf{S4.2} \cite{articleS42}. The possibility to simulate classical epistemic modalities via \BBI-extensions could be a new approach to the field of substructural epistemic logics that consider the notion of evidence \cite{article_Substructural_Epistemic}.

\bibliographystyle{plain}
\bibliography{References}

\appendix

\section{Extended Proofs}
\subsection{Derivability Under Assumptions}

\begin{proof}[Proof of Lemma~\ref{lemma_HS4_b}(Lemma for HS4 Completeness)]
 We prove the result by induction on the structure of the derivation in \SFour.
\begin{description}
    \item[1.] $\vdash_{\SFour}\bigwedge\Gamma \rightarrow \alpha$ is an axiom. Thus, we apply rule 2 of \textbf{HS4} for void $\Gamma$, and obtain $\vdash_{\textbf{HS4}}\bigwedge\Gamma \rightarrow \alpha$. We obtain the desired result (modulo the formula translation) by detachment. It can be easily proved by induction that $\Gamma \vdash_{\textbf{HS4}}\alpha$ iff $\bigwedge \Gamma \vdash_{\textbf{HS4}}\alpha$: we can always reduce reasoning in a multiset of premises to reasoning made in a multiset of premises containing only the conjunction of all the formulae of the original multiset.

\item[2.] $\vdash_{\SFour}\bigwedge\Gamma \rightarrow \alpha$ is the result of application of \textit{modus ponens}. Thus, we have also in \SFour $\vdash_{\SFour}\alpha'$ and $\vdash_{\SFour}\alpha' \rightarrow (\bigwedge \Gamma \rightarrow \alpha)$. The second formula can be easily proven equivalent to $\vdash_{\SFour}\bigwedge \Gamma \rightarrow (\alpha' \rightarrow \alpha)$. Out of Lemma~\ref{lemma_HS4_a}, we obtain $\vdash_{\textbf{HS4}}\bigwedge \Gamma \rightarrow (\alpha' \rightarrow \alpha)$. Then, we apply detachment to it, obtaining $\bigwedge \Gamma \vdash_{\textbf{HS4}} \alpha' \rightarrow \alpha$. We apply the induction hypothesis on $\vdash _{\SFour}\alpha'$, obtaining $\vdash _{\textbf{HS4}}\alpha'$. To this, we apply rule 3 with $\bigwedge \Gamma \vdash_{\textbf{HS4}} \alpha' \rightarrow \alpha$, obtaining the result (modulo formula translation) of $\bigwedge \Gamma \vdash_{\textbf{HS4}} \alpha$.

\end{description}   
\end{proof}

\begin{proof}[Proof of Lemma~\ref{lemma_HBBI_a}]
    We proceed by induction on the structure of the derivation of $\alpha$ in \BBI.
 \begin{description}
        \item[1.] If $\alpha$ is an axiom, then we apply rule 2/rule 4 of \textbf{HBBI}/\textbf{HBBI*} for void $\Gamma$.

        \item[2.] We go case by case on the rules of \BBI.

\begin{description}

\item[2.1.] $\alpha$ is the result of the application of \textit{modus ponens}. Thus, we also have derivations in \BBI of the premises $\beta$ and $\beta \rightarrow \alpha$. By induction hypothesis on both, we have $\vdash_{\textbf{HBBI}} \beta$ and $\vdash_{\textbf{HBBI}} \beta \rightarrow \alpha$. Then, we apply rule 5 of \textbf{HBBI} (in which we note both contexts are void), and obtain $\vdash_{\textbf{HBBI}} \beta$. In \textbf{HBBI*} the proof is analogous.

\item[2.2.] $\alpha$ is the result of the application of rule R2 of \BBI. Thus, $\alpha \equiv (\beta * \gamma) \rightarrow \delta$, and we have in standard axiomatic \BBI that $\vdash_{\BBI}\beta \rightarrow (\gamma \magicwand \delta)$. We simply apply the induction hypothesis on the premise, obtain $\vdash_{\textbf{HBBI}/\textbf{HBBI*}} \beta \rightarrow (\gamma \magicwand \delta)$ and apply rule 5/rule 7 of \textbf{HBBI/HBBI*} (which retains the void $\Gamma$).

\item[2.3.] $\alpha$ is the result of the application of rule R3 of \BBI. The case is similar to the previous one.

\item[2.4.] $\alpha$ is the result of the application of rule R4 of \BBI. Thus, $\alpha \equiv (\beta * \gamma) \rightarrow (\delta * \epsilon)$, and we have premises $\vdash_{\BBI}\beta \rightarrow \delta$ and $\vdash_{\BBI}\gamma \rightarrow \epsilon$. We apply the induction hypothesis on both, obtaining $\vdash_{\textbf{HBBI}/\textbf{HBBI*}} \beta \rightarrow \gamma$ and $\vdash_{\textbf{HBBI}/\textbf{HBBI*}} \delta \rightarrow \epsilon$, and then apply for void $\Gamma$ rule 4/rule 6 of \textbf{HBBI}/\textbf{HBBI*} to obtain $\vdash_{\textbf{HBBI}/\textbf{HBBI*}} (\beta * \delta) \rightarrow (\gamma * \epsilon)$.

\end{description}
\end{description}
\end{proof}

Before proving Lemma~\ref{HBBI_Completeness_b1}, we prove two additional lemmas in standard axiomatic \BBI. These results are used for the rules of \textbf{HBBI*} relative to weakening and contraction in bunches.

\begin{lemma}[Lemma for Weakening]\label{HBBI_Weakening_lemma} If $\vdash_{\BBI}\mathcal{T}(\Gamma[\Gamma'])\rightarrow \alpha$, then $\vdash_{\BBI}\mathcal{T}(\Gamma[\Gamma'; \Gamma''])\rightarrow \alpha$.
\end{lemma}

\begin{proof}
     We proceed by induction on the bunch $\Gamma[-]$.

\begin{description}
    \item[1.] For the base case, we assume that “$\Gamma[-]$'' in the premise is the empty bunch “$[-]$''. Thus, the premise is of the form $\vdash_{\BBI}\mathcal{T}(\Gamma')\rightarrow \alpha$.
    Given its classical basis, it can be axiomatically proven in \BBI that $(\alpha_1 \rightarrow \alpha_2) \rightarrow ((\alpha_1 \land \alpha_3) \rightarrow \alpha_2)$, and thus we obtain the desired $\vdash_{\BBI}\mathcal{T}(\Gamma')\land\mathcal{T}(\Gamma'') \rightarrow \alpha$, from \textit{modus ponens} and the right instantiation of such a scheme.
    
\item[2.] The bunch “$\Gamma[-]$'' has as its outermost structural symbol the “$,$''. Thus, the premise is of the form $\vdash_{\BBI}\mathcal{T}(\Delta'[\Gamma'], \Delta'')\rightarrow \alpha$. For definition of the formula translation of bunches, this is equivalent to $\vdash_{\BBI}(\mathcal{T}(\Delta'[\Gamma']) * \mathcal{T(}\Delta''))\rightarrow \alpha$. By applying rule $[\magicwand2]$ of \BBI on it, we obtain $\vdash_{\BBI}\mathcal{T}(\Delta'[\Gamma']) \rightarrow (\mathcal{T}(\Delta'')\magicwand \alpha)$. We apply the induction hypothesis, obtaining $\vdash_{\BBI}\mathcal{T}(\Delta'[\Gamma'; \Gamma'']) \rightarrow (\mathcal{T}(\Delta'')\magicwand \alpha)$. By applying the rule $[\magicwand2]$, we obtain $\vdash_{\BBI}(\mathcal{T}(\Delta'[\Gamma'; \Gamma'']) * \mathcal{T}(\Delta''))\rightarrow \alpha$, which is the desired $\vdash_{\BBI}\mathcal{T}(\Delta'[\Gamma'; \Gamma''],\Delta'')\rightarrow \alpha$.

\item[3.] The bunch “$\Gamma[-]$'' has as its outermost structural symbol the “$;$''. Thus, the premise is of the form $\vdash_{\BBI}\mathcal{T}(\Delta'[\Gamma']; \Delta'')\rightarrow \alpha$. For the definition of the formula translation of bunches, this is equivalent to $\vdash_{\BBI}(\mathcal{T}(\Delta'[\Gamma']) \land \mathcal{T(}\Delta''))\rightarrow \alpha$, which is equivalent to $\vdash_{\BBI}(\mathcal{T}(\Delta'[\Gamma']) \rightarrow (\mathcal{T}(\Delta'')\rightarrow \alpha)$. We apply the induction hypothesis, and obtain $\vdash_{\BBI}(\mathcal{T}(\Delta'[\Gamma', \Gamma'']) \rightarrow (\mathcal{T}(\Delta'')\rightarrow \alpha)$. For classical logic, this is equivalent to $\vdash_{\BBI}(\mathcal{T}(\Delta'[\Gamma', \Gamma'']) \land (\mathcal{T}(\Delta''))\rightarrow \alpha$, which is in fact the desired $\vdash_{\BBI}\mathcal{T}(\Delta'[\Gamma'; \Gamma''];\Delta'')\rightarrow \alpha$.
    
\end{description}
\end{proof}

\begin{lemma}[Lemma for Contraction]\label{HBBI_Contraction_lemma} If $\vdash_{\BBI}\mathcal{T}(\Gamma[\Gamma'; \Gamma'])\rightarrow \alpha$, then $\vdash_{\BBI}\mathcal{T}(\Gamma[\Gamma'])\rightarrow \alpha$.
    
\end{lemma}

\begin{proof}
 We proceed by induction on the bunch $\Gamma[-]$. For the base case, we assume that “$\Gamma[-]$'' in the premise is the empty bunch “$[-]$''. Thus, the premise is ultimately of the form $\vdash_{\BBI}\mathcal{T}(\Gamma'; \Gamma')\rightarrow \alpha$. By definition of the formula translation of bunches, we have that $\vdash_{\BBI}\mathcal{T}(\Gamma') \land \mathcal{T}(\Gamma')\rightarrow \alpha$, which for classical logic, is equivalent to the desired $\vdash_{\BBI}\mathcal{T}(\Gamma')\rightarrow \alpha$. The rest of the induction cases are analogous to the ones for Lemma~\ref{HBBI_Weakening_lemma}. 
   
\end{proof}

\begin{proof}[Proof of Lemma~\ref{lemma_HS4_b}]
The theorem deals with the translation of multisets/bunches of formulae into standard axiomatic \BBI. We proceed by induction on derivation $\Gamma \vdash_{\textbf{HBBI}/\textbf{HBBI*}} \alpha$ first in \textbf{HBBI}, then in \textbf{HBBI*}: thus, we will omit the subscript for “$\vdash$''.

    \item[] For \textbf{HBBI}:

\begin{description}

    \item[1.] $\Gamma \vdash \alpha$ is the result of application of rule 1 or 2 of \textbf{HBBI}, and thus is either an axiom or belongs to $\Gamma$. If it's an axiom, because \BBI shares the same axioms of \textbf{HBBI}, we have $\alpha$ in the standard Hilbert calculus. Moreover, $\alpha \rightarrow (\bigwedge \Gamma \rightarrow \alpha)$ is an axiom. By \textit{modus ponens} of the standard Hilbert calculus on $\alpha$ and $\alpha \rightarrow (\bigwedge \Gamma \rightarrow \alpha)$, we have $\vdash_{\BBI}\bigwedge \Gamma \rightarrow \alpha$. If $\alpha$ is in $\Gamma$, then for some $\gamma_i \in \Gamma:=\{\gamma_1, ..., \gamma_n\}$, $\alpha\equiv \gamma_i$. Then, $\gamma_i \rightarrow (\gamma_1 \rightarrow \gamma_i)$ is an axiom, and with \textit{modus ponens} with axiom $(\gamma_i \rightarrow (\gamma_1 \rightarrow \gamma_i)) \rightarrow (\gamma_1 \rightarrow (\gamma_i \rightarrow \gamma_i))$, we have $\gamma_1 \rightarrow (\gamma_i \rightarrow \gamma_i)$ in the standard Hilbert calculus. Thus, we continue to weaken for the remaining formulae of $\Gamma$, ultimately obtaining $\gamma_n \rightarrow(\gamma_{n-1}\rightarrow (... \gamma_1 \rightarrow (\gamma_i \rightarrow \gamma_i)))$. Thus, we have that $\vdash_{\BBI} \bigwedge \Gamma \rightarrow \gamma_i$ in standard Hilbert calculus.

    \item[2.] $\Gamma \vdash \alpha$ is the result of application of rule 3 of \textbf{HBBI}. Thus, $\Gamma \vdash \alpha \equiv \Gamma; \Delta \vdash \alpha$, with $\Gamma \vdash \alpha'$ and $\Delta \vdash \alpha' \rightarrow \alpha$ premises of the rule.  We apply the induction hypothesis on both premises, and respectively obtain $\mathcal{T}(\Gamma) \rightarrow \alpha'$ and $\mathcal{T}(\Delta)\rightarrow(\alpha' \rightarrow \alpha)$ in standard \BBI. Considering that it can be axiomatically proven in classical logic, and thus in \BBI, that $((\alpha_1 \rightarrow \alpha_2)\land (\alpha_3 \rightarrow (\alpha_1 \rightarrow \alpha_4))) \rightarrow ((\alpha_1 \land \alpha_3) \rightarrow \alpha_4)$, we thus obtain $\vdash_{\BBI}(\bigwedge \Gamma') \land (\bigwedge \Gamma'') \rightarrow \alpha$ in the standard Hilbert system.

    \item[3.] $\Gamma \vdash \alpha$ is the result of application of rule 4 of \textbf{HBBI}. Thus, $\Gamma \vdash \alpha$ is of form $\Gamma \vdash (\alpha_1 * \alpha_2) \rightarrow (\alpha_3 * \alpha_4)$, with premises $\vdash \alpha_1 \rightarrow \alpha_3$ and $\vdash \alpha_2 \rightarrow \alpha_4$. By applying the induction hypothesis on the premises, we gain in standard \BBI both $\vdash_{\BBI}\top \rightarrow (\alpha_1 \rightarrow \alpha_3)$ and $\vdash_{\BBI}\top \rightarrow (\alpha_2 \rightarrow \alpha_4)$, respectively equivalent to $\vdash_{\BBI}(\alpha_1 \rightarrow \alpha_3)$ and $\vdash_{\BBI}(\alpha_2 \rightarrow \alpha_4)$. Then, we apply rule R4 of \BBI, and thus gain $\vdash_{\BBI}(\alpha_1 * \alpha_2)\rightarrow(\alpha_3 * \alpha_4)$, which we weaken to the desired $\vdash_{\BBI}\mathcal{T}(\Gamma)\rightarrow ((\alpha_1 * \alpha_2)\rightarrow(\alpha_3 * \alpha_4))$.

    \item[4.] $\Gamma \vdash \alpha$ is the result of application of rule 5 of \textbf{HBBI}. Thus, $\Gamma \vdash \alpha \equiv \Gamma \vdash (\alpha_1 * \alpha_2) \rightarrow \alpha_3$, with premise in \textbf{HBBI} $\Gamma \vdash \alpha_1 \rightarrow (\alpha_2 \magicwand \alpha_3)$. We apply the induction hypothesis on the premise, and obtain $\vdash_{\BBI}\mathcal{T}(\Gamma) \rightarrow (\alpha_1 \rightarrow (\alpha_2 \magicwand \alpha_3))$. Considering how we can axiomatically prove through R2 and R3 of \BBI that $\vdash_{\BBI}((\alpha_1 * \alpha_2) \rightarrow \alpha_3) \leftrightarrow (\alpha_1 \rightarrow (\alpha_2 \magicwand \alpha_3))$, we obtain $\vdash_{\BBI}\mathcal{T}(\Gamma) \rightarrow ((\alpha_1 * \alpha_2) \rightarrow \alpha_3)$.

    \item[5.] $\Gamma \vdash \alpha$ is the result of application of rule 6 of \textbf{HBBI}, and $\Gamma \vdash \alpha \equiv \Gamma \vdash (\alpha_1 \rightarrow (\alpha_2 \magicwand \alpha_3))$. The case is analogous to the previous one.

\end{description}

\item[]For \textbf{HBBI*}:

\begin{description}

    \item[1.] $\Gamma \vdash \alpha$ is the result of application of rule 1 of \textbf{HBBI*}. Thus, is of form $\alpha \vdash \alpha$. Since $\mathcal{T}(\alpha)= \alpha$, we only need to notice that $\alpha \rightarrow \alpha$ is an axiom, and thus we have $\vdash_{\BBI}\alpha \rightarrow \alpha$.

     \item[2.] $\Gamma \vdash \alpha$ is the result of application of rule 2 of \textbf{HBBI*}. Thus, is of form $\Gamma[\Gamma'; \Gamma''] \vdash \alpha$. The premise of the rule is $\Gamma[\Gamma'] \vdash \alpha$: by applying the induction hypothesis on it, we get $\vdash_{\BBI}\mathcal{T}(\Gamma[\Gamma'])\rightarrow \alpha$. The desired result obtains out of Lemma~\ref{HBBI_Weakening_lemma}.

     \item[3.] $\Gamma \vdash \alpha$ is the result of application of rule 3 of \textbf{HBBI*}. Thus, is of form $\Gamma[\Gamma']\vdash \alpha$, with premise $\Gamma[\Gamma'; \Gamma']\vdash \alpha$. By applying the induction hypothesis on the premise, we get $\vdash_{\BBI}\mathcal{T}(\Gamma[\Gamma'; \Gamma'])\rightarrow \alpha$. The desired result obtains out of Lemma~\ref{HBBI_Contraction_lemma}.

    \item [4.] $\Gamma \vdash \alpha$ is the result of application of rule 4 of \textbf{HBBI*}. Thus, $\alpha$ is among the axioms of \textbf{HBBI}. Considering how the axioms are shared with the standard axiomatic system, we have also $\alpha$ in \BBI. Also, we can instantiate axiom scheme $\beta \rightarrow (\gamma \rightarrow \beta)$ with $\alpha \rightarrow (\mathcal{T}(\Gamma) \rightarrow \alpha)$. Thus, we obtain $\vdash_{\BBI}\mathcal{T}(\Gamma) \rightarrow \alpha$ by \textit{modus ponens}.

    \item[5.] Cases for $\Gamma \vdash \alpha$ obtained from rules 5-8 are analogous to the corresponding rules for \textbf{HBBI}.

     \item[6.] $\Gamma \vdash \alpha$ is the result of rule 9 of \textbf{HBBI*}, and is of form $\Gamma, \Delta \vdash (\alpha * \beta)$, with premises $\Gamma\vdash \alpha$ and $\Delta \vdash \beta$. By applying the induction hypothesis on the premises, we gain $\vdash_{\BBI}\mathcal{T}(\Gamma)\rightarrow \alpha$ and $\vdash_{\BBI}\mathcal{T}(\Delta) \rightarrow \beta$, to which we apply R4 of \BBI to gain the desired $\vdash_{\BBI}\mathcal{T}(\Gamma) * \mathcal{T}(\Delta) \rightarrow (\alpha * \beta)$. 

\end{description}
    
\end{proof}

\begin{proof}[Proof of Theorem~\ref{Deduction_theorem} (Deduction Theorem for \BBI)]
    We proceed by induction on the structure of the given derivation in \textbf{HBBI}:

    \begin{description}
        \item[1.] If $\alpha; \Gamma \vdash_{\textbf{HBBI}} \beta$ is obtained from rule 1 of \textbf{HBBI}, then $\beta$ belongs to the multiset of assumptions $\Gamma, \alpha$. Thus, either $\beta \equiv \alpha$ or $\beta \in \Gamma$. In the first case, we would have to prove $\Gamma \vdash \beta \rightarrow \beta$, which is straightforward considering how $\beta \rightarrow \beta$ is an axiom (and thus the derivation obtains out of rule 2 of \textbf{HBBI}). In the second case, $\beta \in \Gamma$, thus we have $\Gamma \vdash \beta$ out of rule 1 of \textbf{HBBI}. Also, $\beta \rightarrow (\alpha \rightarrow \beta) \in Axiom$, thus $\vdash \beta \rightarrow (\alpha \rightarrow \beta)$ out of rule 2 of \textbf{HBBI} (for void $\Gamma$). Out of rule 3 of \textbf{HBBI}, we then obtain $\Gamma \vdash \alpha \rightarrow \beta$ out of the premises.

        \item[2.] If $\alpha; \Gamma \vdash_{\textbf{HBBI}} \beta$ is obtained from rule 2 of \textbf{HBBI}, we proceed as the previous case.

        \item[3.] If $\alpha; \Gamma \vdash \beta$ is obtained from rule 3 of \textbf{HBBI}, then $\Gamma \equiv \Gamma'; \Gamma''$ and $\alpha$ is either in the antecedent of the first premise or in the antecedent of the second premise: in the first case we have $\alpha;\Gamma' \vdash \gamma$, in the second instead $\alpha;\Gamma'' \vdash \gamma \rightarrow \beta$.
\begin{description}
    \item[3.1] If $\alpha;\Gamma' \vdash \gamma$, then we apply to it the induction hypothesis, and obtain $\Gamma' \vdash \alpha \rightarrow \gamma$. Considering how $(\gamma \rightarrow \beta)\rightarrow((\alpha \rightarrow \gamma) \rightarrow (\alpha \rightarrow \beta)) \in Axiom$, we obtain through rule 2 of \textbf{HBBI} for void $\Gamma$ the formula $\vdash(\gamma \rightarrow \beta) \rightarrow ((\alpha \rightarrow \gamma)\rightarrow(\alpha \rightarrow \beta))$. By applying rule 3 of \textbf{HBBI} first with $\Gamma'' \vdash \gamma \rightarrow \beta$ and then with $\Gamma' \vdash \alpha \rightarrow \gamma$, we obtain $\Gamma \vdash \alpha \rightarrow \beta$.

    \item[3.2] If $\alpha;\Gamma'' \vdash \gamma \rightarrow \beta$, then we apply the induction hypothesis on it, and obtain $\Gamma'' \vdash \alpha \rightarrow (\gamma \rightarrow \beta)$. Out of classical axiom scheme $(\alpha \rightarrow (\beta \rightarrow \gamma)) \rightarrow (\beta \rightarrow (\alpha \rightarrow \gamma))$ (and rule 2 of \textbf{HBBI}, with void $\Gamma$) and rule 3 of \textbf{HBBI}, we get $\Gamma''\vdash \gamma \rightarrow (\alpha \rightarrow \beta)$. Then, we apply rule 3 of \textbf{HBBI} with $\Gamma' \vdash \gamma$ to get the desired $\Gamma \vdash \alpha \rightarrow \beta$.
    
\end{description}
        
\item[4.] If $\alpha; \Gamma \vdash_{\textbf{HBBI}} \beta$ is obtained from rule 4 of \textbf{HBBI}, then the formula is of the shape $\alpha'; \Gamma \vdash (\alpha * \beta)\rightarrow(\gamma * \delta)$, with premises $\vdash \alpha \rightarrow \gamma$ and $\vdash \beta \rightarrow \delta$. We apply rule 4 of \textbf{HBBI} on both premises, and obtain $\Gamma \vdash (\alpha * \beta)\rightarrow(\gamma * \delta)$. Then, we notice that $((\alpha * \beta)\rightarrow(\gamma * \delta)) \rightarrow (\alpha' \rightarrow ((\alpha * \beta)\rightarrow(\gamma * \delta)) \in Axiom$, and thus obtain out of rule 2 of \textbf{HBBI} (with void $\Gamma$) the formula $\vdash ((\alpha * \beta)\rightarrow(\gamma * \delta)) \rightarrow (\alpha' \rightarrow ((\alpha * \beta)\rightarrow(\gamma * \delta))$. Then, out of rule 3 of \textbf{HBBI}, we obtain $\Gamma\vdash\alpha' \rightarrow ((\alpha * \beta)\rightarrow(\gamma * \delta)$

\item[5.] If $\alpha; \Gamma \vdash_{\textbf{HBBI}} \beta$ is obtained from rule 5 of \textbf{HBBI}, then the formula is of the shape $\alpha'; \Gamma \vdash (\alpha * \beta) \rightarrow \gamma$, with premise $\alpha'; \Gamma \vdash (\alpha \rightarrow (\beta \magicwand \gamma))$. We apply the induction hypothesis on the premise, and obtain $\Gamma \vdash \alpha' \rightarrow(\alpha \rightarrow (\beta \magicwand \gamma))$. Considering how it can be proven in \BBI that $(\alpha \rightarrow (\beta \magicwand \gamma))\rightarrow ((\alpha * \beta)\rightarrow \gamma)$, then by Lemma~\ref{HBBI_Completeness_b1} we have $\vdash (\alpha \rightarrow (\beta \magicwand \gamma))\rightarrow ((\alpha * \beta)\rightarrow \gamma)$, and we can easily conclude $\Gamma \vdash \alpha' \rightarrow ((\alpha * \beta) \rightarrow \gamma)$.

\item[6.] The case with rule 6 of \textbf{HBBI} is analogous to the previous one.

    \end{description}
    
    The deduction theorem on $\rightarrow$ for \textbf{HBBI*} is proven in a similar manner. Let's now prove the second result by induction on the derivation in \textbf{HBBI*}: if $\alpha, \Gamma \vdash_{\textbf{HBBI*}} \beta$, then $\Gamma \vdash_{\textbf{HBBI*}} \alpha \magicwand \beta$. 

    \begin{description}

        \item[1.] $\Gamma ,\alpha\vdash \beta$ is obtained from rule 2 of \textbf{HBBI*}. Thus, is of form $\alpha, \Gamma[\Delta'; \Delta'']\vdash \beta$, with premise $\alpha, \Gamma[\Delta'] \vdash \beta$. Thus, we simply apply the induction hypothesis on the premise, obtaining $\Gamma[\Delta'] \vdash \alpha \magicwand \beta$, and then weaken to the desired $\Gamma[\Delta'; \Delta''] \vdash \alpha \magicwand \beta$ through rule 2 of \textbf{HBBI*}.

        \item[2.] $\Gamma ,\alpha\vdash \beta$ is obtained from rule 3 of \textbf{HBBI*}. Thus, is of form $\alpha, \Gamma[\Delta']\vdash \beta$, with premise $\alpha, \Gamma[\Delta'; \Delta'] \vdash \beta$. The case is similar to the previous one.

        \item[3.]  $\Gamma ,\alpha'\vdash \alpha$ is obtained from rule 4 of \textbf{HBBI*}, and thus $\alpha \in Axiom$. We apply rule 4 again, this time obtaining $\Gamma \vdash \alpha$. Also, we notice that $\alpha \rightarrow ((\alpha*\alpha')\rightarrow \alpha) \in Axiom$. Thus, we apply again rule 4, obtaining $\vdash \alpha \rightarrow ((\alpha*\alpha')\rightarrow \alpha)$. Considering how we have $\Gamma \vdash \alpha$ and $\vdash\alpha \rightarrow ((\alpha*\alpha')\rightarrow \alpha)$, we obtain $\Gamma \vdash (\alpha * \alpha') \rightarrow \alpha$ out of rule 5 of \textbf{HBBI*}. We then apply rule 8 of \textbf{HBBI*}, and obtain $\Gamma \vdash \alpha \rightarrow (\alpha' \magicwand \alpha)$. Then, we apply rule 5 between $\Gamma\vdash \alpha$ and $\Gamma \vdash \alpha \rightarrow (\alpha' \magicwand \alpha)$, then contraction through rule 3 of \textbf{HBBI*}, and conclude $\Gamma \vdash \alpha' \magicwand \alpha$.

        \item[4.] $\Gamma ,\alpha\vdash \beta$ is obtained from rule 5 of \textbf{HBBI*}, and thus is of form $\alpha', \Gamma; \Delta \vdash \beta$. We want to prove $\Gamma; \Delta \vdash \alpha' \magicwand \beta$. The premises can be of the following forms:

\begin{description}

\item[4.1] $\alpha', \Gamma \vdash \alpha \rightarrow \beta$ and $\Delta\vdash \alpha$. We apply the induction hypothesis on the first premise, and obtain $\Gamma \vdash \alpha' \magicwand (\alpha \rightarrow \beta)$. We easily prove the \textit{a fortiori} of this formula, obtaining $\Gamma \vdash \alpha \rightarrow (\alpha' \magicwand (\alpha \rightarrow \beta))$. Then, we apply rule 7 of \textbf{HBBI*}, and obtain $\Gamma \vdash (\alpha * \alpha') \rightarrow (\alpha \rightarrow \beta)$. Considering how $(\alpha_1 \rightarrow (\alpha_2 \rightarrow \alpha_3)) \rightarrow (\alpha_2 \rightarrow (\alpha_1 \rightarrow \alpha_3))$ is a classical tautology, and thus can be easily proven in \BBI, by Lemma~\ref{lemma_HBBI_a} we have its instantiation $\vdash((\alpha* \alpha') \rightarrow (\alpha \rightarrow \beta)) \rightarrow (\alpha \rightarrow ((\alpha*\alpha') \rightarrow \beta))$ in \textbf{HBBI*}. We apply rule 5 of \textbf{HBBI*} between $\Gamma \vdash (\alpha * \alpha') \rightarrow (\alpha \rightarrow \beta)$ and $\vdash((\alpha* \alpha') \rightarrow (\alpha \rightarrow \beta)) \rightarrow (\alpha \rightarrow ((\alpha*\alpha') \rightarrow \beta))$, obtaining $\Gamma \vdash(\alpha \rightarrow ((\alpha*\alpha') \rightarrow \beta)$. Then, we apply again rule 5, this time with premise $\Delta \vdash \alpha$, obtaining $\Gamma; \Delta \vdash (\alpha * \alpha') \rightarrow \beta$. Then, we apply rule 8 of \textbf{HBBI*}, obtaining $\Gamma; \Delta \vdash \alpha \rightarrow (\alpha' \magicwand \beta)$. Then, we apply again rule 5 with premise $\Delta \vdash \alpha$, obtaining $\Gamma; \Delta; \Delta \vdash \alpha' \magicwand \beta$. The desired $\Gamma; \Delta \vdash \alpha' \magicwand \beta$ follows by contraction through rule 3.

\item[4.2] $\Gamma \vdash \alpha \rightarrow \beta$ and $\alpha', \Delta\vdash \alpha$. We apply the induction hypothesis on the second premise, and obtain $\Delta \vdash \alpha' \magicwand \alpha$. Again, we easily obtain the \textit{a fortiori} of it, thus having $\Delta \vdash (\alpha \rightarrow \beta) \rightarrow (\alpha' \magicwand \alpha)$, to which we apply rule 7 of \textbf{HBBI*}, obtaining $\Delta \vdash ((\alpha \rightarrow \beta)* \alpha') \rightarrow \alpha$. Considering how in classical logic can be easily proved $(\alpha_1 \rightarrow \alpha_2) \rightarrow ((\alpha_2 \rightarrow \alpha_3)\rightarrow (\alpha_1 \rightarrow \alpha_3))$, we can easily prove in \BBI its instantiation $((\alpha \rightarrow \beta)*\alpha' \rightarrow \alpha)\rightarrow ((\alpha \rightarrow \beta)\rightarrow((\alpha \rightarrow \beta)*\alpha'\rightarrow \beta))$. Out of Lemma~\ref{lemma_HBBI_a} thus, we obtain $\vdash((\alpha \rightarrow \beta)*\alpha' \rightarrow \alpha)\rightarrow ((\alpha \rightarrow \beta)\rightarrow((\alpha \rightarrow \beta)*\alpha'\rightarrow \beta))$ in \textbf{HBBI*}. We apply rule 5 of \textbf{HBBI*} between $\vdash((\alpha \rightarrow \beta)*\alpha' \rightarrow \alpha)\rightarrow ((\alpha \rightarrow \beta)\rightarrow((\alpha \rightarrow \beta)*\alpha'\rightarrow \beta))$ and $\Delta \vdash ((\alpha \rightarrow \beta)* \alpha') \rightarrow \alpha$, obtaining $\Delta\vdash ((\alpha \rightarrow \beta)\rightarrow((\alpha \rightarrow \beta)*\alpha'\rightarrow \beta))$. Then, we apply rule 5 again between $\Delta\vdash ((\alpha \rightarrow \beta)\rightarrow((\alpha \rightarrow \beta)*\alpha'\rightarrow \beta))$ and hypothesis $\Gamma \vdash \alpha \rightarrow \beta$, obtaining $\Gamma; \Delta \vdash (\alpha \rightarrow \beta)*\alpha'\rightarrow \beta$. We apply to it rule 8 of \textbf{HBBI*}, obtaining $\Gamma; \Delta \vdash (\alpha \rightarrow \beta)\rightarrow(\alpha'\magicwand \beta)$. We apply again rule 5 $\Gamma; \Delta \vdash (\alpha \rightarrow \beta)\rightarrow(\alpha'\magicwand \beta)$ and hypothesis $\Gamma \vdash \alpha \rightarrow \beta$, obtaining $\Gamma; \Gamma; \Delta \vdash \alpha' \magicwand \beta$. Applying contraction through rule 3, we obtain the desired $\Gamma; \Delta \vdash \alpha' \magicwand \beta$.
    
\end{description}

        \item[5.] $\Gamma ,\alpha\vdash \beta$ is obtained from rule 6 of \textbf{HBBI*}, and thus is of form $\alpha', \Gamma \vdash (\alpha * \beta) \rightarrow (\gamma * \delta)$. The premises are $\vdash \alpha \rightarrow \gamma$ and $\vdash \beta \rightarrow \delta$. We apply rule 6 to the premises, this time obtaining $\Gamma \vdash (\alpha * \beta) \rightarrow (\gamma * \delta)$. We produce the \textit{a fortiori} of $\Gamma \vdash (\alpha * \beta) \rightarrow (\gamma * \delta)$, obtaining $\Gamma \vdash (((\alpha * \beta) \rightarrow (\gamma * \delta))* \alpha') \rightarrow ((\alpha * \beta) \rightarrow (\gamma * \delta))$. We then apply rule 8 of \textbf{HBBI*}, and obtain $\Gamma \vdash (((\alpha * \beta) \rightarrow (\gamma * \delta)) \rightarrow (\alpha' \magicwand ((\alpha * \beta) \rightarrow (\gamma * \delta)))$. By rule 5 with $\Gamma \vdash (\alpha * \beta) \rightarrow (\gamma * \delta)$ and contraction (through rule 3), we obtain the desired $\Gamma \vdash \alpha' \magicwand ((\alpha * \beta) \rightarrow (\gamma * \delta))$.

        \item[6.] $\Gamma ,\alpha\vdash \beta$ is obtained from rule 7 of \textbf{HBBI*}, and thus is of form 
        $\Gamma, \alpha' \vdash \alpha \rightarrow (\beta \magicwand \gamma)$, with premise $\Gamma, \alpha' \vdash (\alpha * \beta) \rightarrow \gamma$. We apply the induction hypothesis to the premise, and obtain $\Gamma \vdash \alpha' \magicwand ((\alpha * \beta) \rightarrow \gamma)$. Then, we notice that it can be proven in \BBI that $((\alpha * \beta)\rightarrow \gamma)\rightarrow (\alpha \rightarrow (\beta \magicwand \gamma))$. By Lemma~\ref{lemma_HBBI_a}, we have $\vdash((\alpha * \beta)\rightarrow \gamma)\rightarrow (\alpha \rightarrow (\beta \magicwand \gamma))$, and we can easily conclude $\Gamma \vdash \alpha' \rightarrow ((\alpha \rightarrow (\beta \magicwand \gamma))$.

        \item[7.] $\Gamma ,\alpha\vdash \beta$ is obtained from rule 8 of \textbf{HBBI*}, and thus is of form 
        $\Gamma, \alpha' \vdash (\alpha * \beta) \rightarrow \gamma$, with premise $\Gamma, \alpha' \vdash \alpha \rightarrow (\beta \magicwand \gamma)$. The case is similar to the previous one.

        \item[8.] $\Gamma, \alpha\vdash \beta$ is obtained from rule 9 of \textbf{HBBI*}, and thus is of form $ \Gamma, \alpha \vdash \alpha_1 * \alpha_2$. The premises are either $\Gamma\vdash \alpha_1$ and $\alpha \vdash \alpha_2$, or $\alpha \vdash \alpha_1$ and $\Gamma \vdash \alpha_2$. Let's assume it's the first case (the second it's entirely analogous). Since $\alpha \vdash \alpha_2$, then $\alpha \equiv \alpha_2$. We notice that $(\alpha_1 * \alpha_2) \rightarrow (\alpha_1 * \alpha_2)$ is an axiom, thus we have $\Gamma \vdash (\alpha_1 * \alpha_2) \rightarrow (\alpha_1 * \alpha_2)$ out of rule 4 of \textbf{HBBI*}. Then, we apply rule 8, and obtain $\Gamma \vdash \alpha_1 \rightarrow (\alpha_2 \magicwand (\alpha_1 * \alpha_2))$. We apply rule 5 between $\Gamma \vdash \alpha_1 \rightarrow (\alpha_2 \magicwand (\alpha_1 * \alpha_2))$ and hypothesis $\Gamma \vdash \alpha_1$, then we apply contraction through rule 2, and obtain the desired $\Gamma \vdash \alpha_2 \magicwand (\alpha_1 * \alpha_2)$.

    \end{description}

\end{proof}

\begin{proof}[Proof of (\ref{detachment_theorem_magicwand}), Theorem~\ref{detachment} (Detachment Principle on $\magicwand$ on HBBI*)]

To prove (c), we show that the rule of \textit{modus ponens} on $\magicwand$ is admissible in \textbf{HBBI*}. This means that if we have premises $\Gamma\vdash_{\textbf{HBBI*}} \alpha \magicwand \beta$ and $\Delta \vdash_{\textbf{HBBI*}} \alpha$, we can conclude $\Gamma, \Delta\vdash_{\textbf{HBBI*}} \beta$. We notice that this version of \textit{modus ponens} employs ``,'' rather than ``;'' to join the bunches of assumptions of the premises. To show it, we proceed like this: for classical reasoning, we can easily produce the \textit{a fortiori} $\Gamma\vdash_{\textbf{HBBI*}} (\alpha \magicwand \beta) \rightarrow (\alpha \magicwand \beta)$ of the premise $\Gamma\vdash_{\textbf{HBBI*}} \alpha \magicwand \beta$. Then, we apply rule 7 of \textbf{HBBI*}, and obtain $\Gamma\vdash_{\textbf{HBBI*}} ((\alpha \magicwand \beta)* \alpha) \rightarrow \beta $. We then apply rule 9 of \textbf{HBBI*} to the premises $\Delta\vdash_{\textbf{HBBI*}} \alpha$ and $\Gamma\vdash_{\textbf{HBBI*}} \alpha \magicwand \beta$, and obtain $\Gamma, \Delta \vdash_{\textbf{HBBI*}} (\alpha \magicwand \beta)* \alpha$. We then apply rule 5 of \textbf{HBBI*} between  $\Gamma, \Delta \vdash_{\textbf{HBBI*}} \alpha* (\alpha \magicwand \beta)$ and $\Gamma\vdash_{\textbf{HBBI*}} (\alpha *(\alpha \magicwand \beta)) \rightarrow \beta$, apply contraction on $\Gamma$ through rule 3, and finally obtain the desired $\Gamma, \Delta \vdash \beta.$ Thus, detachment on $\magicwand$ for \textbf{HBBI*} obtains out of the admissible rule of \textit{modus ponens} on $\magicwand$ between premise $\Gamma \vdash_{\textbf{HBBI*}}\alpha \magicwand \beta$ and $\alpha \vdash \alpha$, obtained through rule 1 of \textbf{HBBI*}.
   
\end{proof}

\begin{proof}[Proof of Lemma~\ref{HBBI_Completeness_b2}]
  We proceed by induction on the derivation in \BBI.
\begin{description}

    \item[1.] $\vdash_{\BBI}\mathcal{T} (\Gamma) \rightarrow \alpha$ is an axiom. The result obtains similarly as in case 1 of Lemma~\ref{lemma_HS4_b}.

    \item[2.] $\vdash_{\BBI}\mathcal{T} (\Gamma) \rightarrow \alpha$ is the result of application of \textit{modus ponens}. Again, the result is analogous to case 2 of Lemma~\ref{lemma_HS4_b}.

    \item[3.] $\vdash_{\BBI}\mathcal{T} (\Gamma) \rightarrow \alpha$ is the result of application of R2 of \BBI: thus, $\mathcal{T} (\Gamma) \rightarrow \alpha \equiv (\tau' * \tau'') \rightarrow \alpha$. As a premise, we have $\vdash_{\BBI} \tau' \rightarrow (\tau'' \magicwand \alpha)$. We apply the induction hypothesis on the premise, and obtain $\tau'\vdash_{\textbf{HBBI*}}  \tau'' \magicwand \alpha$. Then, we apply detachment on $\magicwand$, obtaining $\tau', \tau''\vdash_{\textbf{HBBI*}} \alpha$. Considering the definition of the formula translation of bunches, we in fact obtained $\Gamma \vdash \alpha$, with bunch $\Gamma \equiv \tau', \tau''$. 

    \item[4.] $\vdash_{\BBI}\mathcal{T} (\Gamma) \rightarrow \alpha$ is the result of application of R3 of \BBI: thus, $\mathcal{T} (\Gamma) \rightarrow \alpha \equiv \mathcal{T} (\Gamma) \rightarrow (\alpha' \magicwand \alpha'')$. As a premise, we have $\vdash_{\BBI} (\mathcal{T} (\Gamma) *\alpha') \rightarrow \alpha''$. We apply the induction hypothesis on it, and obtain $\mathcal{T}(\Gamma)* \alpha' \vdash_{\textbf{HBBI*}} \alpha''$. We notice that $\mathcal{T}(\Gamma)* \alpha' \vdash_{\textbf{HBBI*}} \alpha''$ is actually equivalent to $\mathcal{T}(\Gamma), \alpha' \vdash_{\textbf{HBBI*}} \alpha''$. Thus, we apply the deduction theorem on $\magicwand$, obtaining $\mathcal{T}(\Gamma) \vdash_{\textbf{HBBI*}}\alpha' \magicwand \alpha''$. Similarly to our reasoning on multisets, having as a premise bunch $\Gamma$ or its formula translation $\mathcal{T}(\Gamma)$ is equivalent. Thus, we easily conclude $\Gamma \vdash_{\textbf{HBBI*}}\alpha' \magicwand \alpha''$

    \item[5.] $\vdash_{\BBI}\mathcal{T} (\Gamma) \rightarrow \alpha$ is the result of application of R4 of \BBI: thus, $\mathcal{T} (\Gamma) \rightarrow \alpha \equiv (\tau'*\tau'') \rightarrow (\alpha' * \alpha'')$. As premises, we have $\vdash_{\BBI}\tau' \rightarrow \alpha'$ and $\vdash_{\BBI}\tau'' \rightarrow \alpha''$. We apply the induction hypothesis to both, obtaining $\tau'\vdash_{\textbf{HBBI*}} \alpha'$ and $\tau''\vdash_{\textbf{HBBI*}} \alpha''$. We apply rule 9 of \textbf{HBBI*}, and obtain the desired $\tau', \tau''\vdash_{\textbf{HBBI*}} \alpha'*\alpha''$.

\end{description}

\end{proof}

\subsection{Extending the Embedding: Hybrid and BiBBI Logics}

For hybrid bunched implication logic, we extend the notion of evaluation and validity in a relational frame by adding the following clauses, as defined in \cite{separation_theories}:

\begin{definition}[Hybrid \BBI valuation] A hybrid evaluation $\rho$ for a \BBI relational resource model extends a standard evaluation by adding a mapping from nominals $l$ to elements of $\rho(l)\in M$. This time, we explicitly notify the presence of $\rho$ in the forcing relationship, as it will be relevant for the clause for the binder $\downarrow$ (should the hybrid language employ that as well). Also, the forcing relationship is extended as follows:

\begin{description}
    \item[] $r\Vdash_\rho l$ iff $\rho(l)=r$.
    \item[] $r \Vdash_\rho @_l \alpha$ iff $\rho(l)\Vdash \alpha$.
    \item[] $r\Vdash_\rho \downarrow l. \alpha$ iff $r \Vdash_{\rho[l:=r]} \alpha$, where $\rho[l:=r]$ is notation for the hybrid evaluation defined as $\rho$ except that $[\rho:=r](l)=_{Def}r$.
\end{description}

In short, the forcing relation evaluates nominal $l$ true at evaluation point $r$ if the reference of $l$ assigned by $\rho$ is $r\in M$, and $@_l\alpha$ if the formula $\alpha$ is true at the evaluation point which is the reference of $l$. Binder $\downarrow$ binds a nominal to the point of evaluation, and thus $\downarrow l. \alpha$ is true at evaluation point $r$ if and only if $\alpha$ is true at $r$ when $l$ refers to $r$.
    
\end{definition}

We recall that system \textbf{K}$_{\mathcal{H}(@)}$ is proven complete with regard to arbitrary extensions with pure axioms \cite{Hybrid_Logic_Pure_extensions}.
Thus, in order to obtain a sound and complete axiomatic system for hybrid logic \SFour, we add to \textbf{K}$_{\mathcal{H}(@)}$ pure axioms expressing reflexivity and transitivity: these are axioms 8 and 9 of the following definition.

\begin{definition}[Axiomatic Hybrid \SFour]
We add to a finite set of axioms for propositional classical logic the following axiom schemes:

\begin{enumerate}
\item [1.] $\vdash\Box(\alpha \rightarrow \beta) \rightarrow(\Box \alpha \rightarrow \Box\beta)$.
\item[2.] $\vdash@_i(\alpha \rightarrow \beta) \rightarrow(@_i \alpha \rightarrow @_i\beta)$.
\item[3.] $\vdash@_i\alpha \leftrightarrow \neg@_i \neg \alpha$.
\item[4.] $\vdash@_i i$.
\item[5.] $\vdash@_i@_j \alpha \leftrightarrow@_j\alpha$.
\item[6.] $\vdash i \rightarrow (\alpha \leftrightarrow@_i \alpha).$
\item[7.] $\vdash \Diamond @_i \alpha \rightarrow@_i \alpha$.
\item[8.] $\vdash @_i\Diamond i$.
\item[9.] $\vdash @_i\Diamond j \land @_j\Diamond k \rightarrow @_i\Diamond k$.

\end{enumerate}

We have the following deduction rules:

\begin{enumerate}

\item[R1.] If $\vdash\alpha$ and $\vdash\alpha \rightarrow \beta$, infer $\vdash\beta$. \textit{(Modus Ponens).}

\item[R2] If $\vdash \alpha$, then $\vdash \alpha^{\sigma}$. \textit{(Substitution).}

\item[R3] If $\vdash\alpha$, then $\vdash\Box \alpha$. \textit{(Generalization-$\Box$).}

\item[R4] If $\vdash\alpha$, then $\vdash@_i \alpha$. \textit{(Generalization-$@$).}

\item[R5] If $\vdash@_i\alpha$ and $i$ does not occur in $\alpha$, then $\vdash \alpha$. \textit{(Name).}

\item[R6] If $\vdash@_i \Diamond j \rightarrow @_j\alpha$ and $j\cancel{=}i$ does not occur in $\alpha$, then $\vdash @_i\Box\alpha$. \textit{(Bounded Generalisation).}

\end{enumerate}
By $\sigma$, we mean any substitution that uniformly replaces nominals with nominals.

\end{definition}

Axiomatic system \textbf{HyBBI} is proposed in \cite{separation_theories}, and is a conservative extension of \BBI. Out of practicality for some technical development, axiomatization of \textbf{HyBBI} in \cite{separation_theories} employs septraction $\septraction$ rather than $\magicwand$.

\begin{definition}[Axiomatic System for Hybrid \BBI] We add to the set of axioms of \BBI the following axiom schemes:

\begin{enumerate}
\item[1.] $\vdash@_i(\alpha \rightarrow \beta) \rightarrow(@_i \alpha \rightarrow @_i\beta)$.
\item[2.] $\vdash@_i\alpha \leftrightarrow \neg@_i \neg \alpha$.
\item[3.] $\vdash i \rightarrow (\alpha \rightarrow@_i \alpha).$
\item[4.] $\vdash@_i i$.
\item[5.] $\vdash@_i k \rightarrow @_k i$.
\item[6.] $\vdash@_i k \land @_k \alpha \rightarrow @_i \alpha$.
\item[7.] $\vdash@_i@_j \alpha \leftrightarrow@_j\alpha$.

\item[8.] $\vdash (@_i (k*k')\land @_k\alpha\land@_{k'}\beta) \rightarrow @_i(\alpha*\beta)$.

\item[9.] $\vdash (@_i (k\septraction k')\land @_k\alpha\land@_{k'}\beta) \rightarrow @_i(\alpha \septraction \beta)$.

\end{enumerate}

A number of axioms in the system are the same (or slight variations) of those of \SFour$_{\mathcal{H}(@)}$. In addition to R1-R5 of \SFour$_{\mathcal{H}(@)}$, we have the following substructural deduction rules:

\begin{enumerate}

\item[SR1.] If $\vdash (@_i(k*k')\land @_k\alpha\land@_{k'}\beta)\rightarrow \gamma$ and $k, k'$ are not in $\alpha, \beta, \gamma$ or $\{i\}$, then $@_i(\alpha * \beta)\rightarrow \gamma$. \textit{(Paste $*$).}

\item[SR2] If $\vdash (@_i(k \septraction k')\land @_k\alpha\land@_{k'}\beta)\rightarrow \gamma$ and $k, k'$ are not in $\alpha, \beta, \gamma$ or $\{i\}$, then $@_i(\alpha \septraction \beta)\rightarrow \gamma$. \textit{(Paste $\magicwand$).}

\end{enumerate}

\end{definition}

We prove that \SFour$_{\mathcal{H}(@)}$ embeds into \textbf{HyBBI}. This result establishes a parametric embedding connecting extensions of \SFour$_{\mathcal{H}(@)}$ to extensions of \textbf{K}$_{HyBBI}$. The translation from the language of hybrid \SFour into the language of hybrid \BBI adds clauses for nominals and the validity operator $@_i$. In compliance with the original notation of both systems, we also directly translate ``$\Diamond$'' in $t(-)$ and ``$\septraction$'' in $\rev{-}$:

\begin{definition}[Translation from \SFour$_{\mathcal{H}(@)}$ into HyBBI] The translation from the language of \SFour$_{\mathcal{H}(@)}$ into the language of \BBI is $t(-): \mathcal{L}_{S4_{\mathcal{H}(@)}}\longrightarrow {\mathcal{L}_{\textbf{HyBBI}}}$ such that

\begin{description}
    \item[1.] $t(\Box \alpha) =_{Def}\top \magicwand t(\alpha)$.
    \item[2.] $t(\Diamond \alpha) = =_{Def}\top \septraction t(\alpha) $
    \item[3.] $t(\neg \alpha) =_{Def}\neg t(\alpha) $.
    \item[4.] $t(\alpha \circ \beta)=_{Def}t(\alpha) \circ t(\beta) $, for $\circ \in \{\land, \lor, \rightarrow\}$.
    \item[5.] $t(C)=_{Def}C$, for $C \in \{\top, \bot\}$ or $C \equiv p$ for $p$ propositional letter or $C\equiv l$ for $l$ nominal.
    \item[6.] $t(@_i\alpha)=_{Def}@_it(\alpha)$ 
    
\end{description}
    
\end{definition}

\begin{definition}[Reverse Translation]
The translation from the language of \textbf{HyBBI} into the language of $\SFour_{\mathcal{H}(@)}$ $\rev{-}:\mathcal{L_{\textbf{HyBBI}}}\longrightarrow {L_{\SFour_{\mathcal{H}(@)}}}$ is the same as that of Definition~\ref{Def:Reverse_Translation}, adding:

\begin{description}
    \item[1.] $[\alpha \septraction \beta]'=_{Def} \alpha \land \beta \land \Diamond \beta$.
    \item[2.] $[@_i\alpha]'=_{Def}@_i[\alpha]'$.
    \item[3.] $[l]'=_{Def}l$ for $l$ nominal. 
\end{description}

\end{definition}

Out of practicality, we employ semantic reasoning rather than axiomatic reasoning throughout the proof. Still, the completeness of both proof systems guarantees the existence of proper derivations, allowing us to obtain the result through purely syntactical means. 

\begin{theorem}[Soundness of $t(-)$ Translation]\label{Theor:Soundness_HyBBI}
If $\vdash_{\SFour_{\mathcal{H}(@)}}\alpha$, then $\vdash_{\textbf{HyBBI}}t(\alpha)$.    
\end{theorem}

\begin{proof}
 We proceed by induction on the length of the formulae. A number of axioms and rules are shared between the systems, so we focus on the ones not present in \textbf{HyBBI}.

\begin{description}
    \item[1.] Axiom $\vdash_{\SFour_{\mathcal{H}(@)}} i \rightarrow (\alpha \leftrightarrow@_i \alpha)$ is left unchanged by translation $t(-)$ into \textbf{HyBBI}. Thus, considering that $\vdash_{\textbf{HyBBI}} i \rightarrow (\alpha \rightarrow@_i \alpha)$ is an axiom, we only need to prove $\vdash_{{_{HyBBI}}} i \rightarrow (@_i \alpha\rightarrow\alpha)$. We work by indirect reasoning, assuming the negation of the statement and obtaining a contradiction. Thus, we assume that if $s\Vdash i$, then $s\cancel{\Vdash}@_i\alpha\rightarrow \alpha$. Out of the semantic clause for nominals, we have that $s\Vdash i$ iff $s=i$. Also, we have that if $i\Vdash @_i\alpha$, then $i\cancel{\Vdash}\alpha$. But $i\Vdash @_i\alpha$ iff $i\Vdash \alpha$. Hence, a contradiction.

    \item[2.] Axiom $\vdash_{\SFour_{\mathcal{H}(@)}} \Diamond@_i\alpha \rightarrow @_i\alpha$ maps into $\textbf{HyBBI}$-formula ``$(\top \septraction @_i\alpha)\rightarrow @_i \alpha$'', which we need to show valid in $\textbf{HyBBI}$. Again, we work by indirect reasoning. We assume that if $s\Vdash \top \septraction @_i\alpha$, then $s\cancel{\Vdash}@_i\alpha$, and obtain a contradiction. By semantic clause, $s\Vdash \top \septraction @_i\alpha$ iff $\exists s',\exists s''$ such that $s,s'\triangleright s''$, and $s'\Vdash \top$ and $s''\Vdash @_i\alpha$. Also, $s''\Vdash @_i\alpha$ iff $i\Vdash \alpha$. Still, $s\cancel{\Vdash}@_i\alpha$ iff $i\cancel{\Vdash}\alpha$. Hence, a contradiction.

    \item[3.] We show that the translation of R6 of bounded generalization of ${\SFour_{\mathcal{H}(@)}}$ is valid in $\textbf{HyBBI}$. The translation of the rule states that if $\vdash_{{\textbf{HyBBI}}} @_i(\top \septraction j)\rightarrow @_j\alpha$ and $j\cancel{=}i$ does not occur in $\alpha$, then we can conclude $\vdash_{{\textbf{HyBBI}}} @_i(\top \magicwand \alpha)$. We assume the premise of the translated rule. By semantic clause, $s\Vdash @_i(\top \septraction j)$ iff $\exists s' \exists s''$ such that $i,s'\triangleright s''$, and $s'\Vdash \top$ and $s''\Vdash j$; thus, considering $\Vdash @_j \alpha$ and the semantic clause on nominals, we have that $s''=j$. Also, we have that $j\Vdash \alpha$. Considering the restriction that nominals $j\cancel{=}i$ do not occur in $\alpha$, $j$ is in fact a generic world seen by $i$ for which we have that $\alpha $obtains. Thus, we may universally quantify on it, and obtain that for all $j$ such that $i,s'\triangleright j$ (and the void condition $s'\Vdash \top$), $j\Vdash \alpha$. Thus, the rule remains valid.
    
\end{description}
   
\end{proof}

We prove the needed properties on the reverse translation.

\begin{lemma}[Soundness of ${[}-{]}'$ Translation]\label{Theor:Faithfulnes_1_HyBBI}
For every formula $\alpha$ of $\textbf{HyBBI}$, if $\vdash_{\textbf{HyBBI}} \alpha$ then $\vdash_{\SFour_{\mathcal{H}(@)}} [\alpha]'$.

\end{lemma}
\begin{proof}
  A number of axioms and rules are shared between systems. For readability, we omit the $\rev{-}$ on sub-formulae $\alpha, \beta$, and so on. We focus on the axioms and rules for substructural operators:

    \begin{description}
        \item[1.] Axiom 8. of $\textbf{HyBBI}$ translates via $\rev{-}$ into $\vdash (@_i (k\land k')\land @_k\alpha\land@_{k'}\beta) \rightarrow @_i(\alpha \land \beta)$. Semantically, this means that we need to prove that if $i\Vdash k \land k'$ and $k \Vdash \alpha$ and $k'\Vdash \beta$, then $i\Vdash \alpha \land \beta$. This easily checks, considering how the semantic clauses imply that $i\Vdash k \land k'$ iff $i\Vdash k$ and $i \Vdash k'$, and thus $i=k=k'$; moreover, the two other conjuncts of the hypothesis state that $k\Vdash \alpha$ and $k'\Vdash \beta$. Thus, $i \Vdash \alpha \land \beta$.

        \item[2.] Axiom 9. of $\textbf{HyBBI}$ translates via $\rev{-}$ into $\vdash (@_i (k\land k' \land \Diamond k')\land @_k\alpha\land@_{k'}\beta) \rightarrow @_i(\alpha \land \beta \land \Diamond \beta)$. Assuming the antecedent, we show the consequent. In fact, $s\Vdash @_i (k\land k'\land \Diamond k')$ iff $i\Vdash k$ and $i\Vdash k'$ and $i\Vdash \Diamond k'$; thus, we have that $i=k=k'$, and that $i\Vdash \Diamond k'$. In short then, world $k$ sees itself. Also, the other conjuncts of the antecedent state that $k\Vdash \alpha$ and $k'\Vdash \beta$. From this, we conclude that $i\Vdash \alpha$ and $i \Vdash \beta$ and $i\Vdash \Diamond \beta$ (this last conjunct obtains because $i=k=k'$ sees itself). Hence, the consequent.

        \item[3.] Rule SR1 of $\textbf{HyBBI}$ translates via $\rev{-}$ into a rule stating that if $@_i(k\land k')\land @_k\alpha \land @_{k'}\rightarrow \gamma$ and $k, k'$ do not occur in $\alpha, \beta,\gamma$ or $\{i\}$, then $@_i(\alpha \land \beta) \rightarrow \gamma$. Assuming the antecedent, we prove the consequent. Thus, if $i\Vdash k \land k'$ and $k\Vdash \alpha$ and $k'\Vdash \beta$, then $s \Vdash \gamma$; hence, $i=k=k'$ and $i\Vdash \alpha$ and $i\Vdash \beta$. Thus, the semantic clause for the consequent (i.e. if $i\Vdash \alpha \land \beta$ then $s\Vdash \gamma$) is easily established. Also, we notice that we had no need to enforce the particular restriction on nominals, which we may assume \textit{a fortiori}.

        \item[4.] Rule SR2 of $\textbf{HyBBI}$ translates via $\rev{-}$ into a rule stating that if $\vdash (@_i(k \land k' \land \Diamond k')\land @_k\alpha\land@_{k'}\beta)\rightarrow \gamma$ and $k, k'$ are not in $\alpha, \beta, \gamma$ or $\{i\}$, then $@_i(\alpha \land \beta \land \Diamond \beta)\rightarrow \gamma$. Assuming the antecedent, we prove the consequent. Unpacking the semantic clauses for the antecedent, we have that if $i\Vdash k \land k' \land \Diamond k'$ and $k\Vdash \alpha$ and $k'\Vdash \beta$, then $s\Vdash \gamma$. Again, we have that $i=k=k'$, and $k'$ sees itself. From these, the semantic clause for the consequent (that is, if $i\Vdash \alpha \land \beta \land \Diamond \beta$ then $s\Vdash \gamma$) is easily established. Again, we had no need to enforce the particular restriction on nominals (which we assume \textit{a fortiori}). In fact, we notice that the reverse translation ``collapses'' worlds that would be separated in the original substructural environment.

    \end{description}
   
\end{proof}

\begin{lemma}[Translation ${[}-{]}'$ is Cancelling]\label{Theor:Faithfulnes_2_HyBBI}
For every formula $\alpha$ in $\SFour_{\mathcal{H}(@)}$, $\rev{t(\alpha)}$ is $\SFour_{\mathcal{H}(@)}$-equivalent to $\alpha$: that is, $\vdash_{\SFour_{\mathcal{H}(@)}} \alpha \leftrightarrow \rev{t(\alpha)}$.

\end{lemma}
\begin{proof}
It can be easily noticed that in the case of axiom $\vdash \Diamond @_i \alpha \rightarrow@_i \alpha$, the only one presenting an interaction between modal and hybrid operators, applying $t(-)$ and then $[-]'$ results in the valid $\vdash_{\SFour_{\mathcal{H}(@)}} (\Diamond @_i \alpha \rightarrow @_i \alpha) \leftrightarrow ((\top \land  @_i \alpha \land \Diamond @_i \alpha) \rightarrow@_i \alpha)$. Considering how translation $\rev{-}$ leaves unaltered the new elements of the language, the rest of the proof is straightforward.    
\end{proof}

Theorem~\ref{Hybrid_Embedding} follows from the aforementioned results in the same way as Theorem~\ref{Theor:Faithfulness}.

\end{document}